\documentclass{article}
\usepackage[utf8]{inputenc}
\usepackage{amsmath}
\usepackage{amssymb}
\usepackage{hyperref}
\usepackage{comment}

\usepackage{pifont}
\newcommand{\cmark}{\ding{51}}%
\newcommand{\xmark}{\ding{55}}%

\usepackage[a4paper]{geometry}
\usepackage{algpseudocode}
\usepackage{xcolor}

\usepackage{listings}
\usepackage{sourcecodepro}

\usepackage{xcolor}
\usepackage{graphicx}
\usepackage{rotating}

\lstdefinestyle{Rstyle}{
    language=R,
    basicstyle=\ttfamily\small,
    commentstyle=\color{gray},
    numbers=left,
    numberstyle=\tiny\color{gray},
    stepnumber=1,
    showstringspaces=false,
    tabsize=4,
    frame=single,
    breaklines=true,
    breakatwhitespace=true,
}

\usepackage{mathtools} 

\usepackage{lipsum}

\usepackage{amsthm}
\newtheorem{theorem}{Theorem}[section]

\newtheorem{definition}{Definition}[section]
\newtheorem{example}{Example}[section]
\newtheorem{lemma}{Lemma}[section]
\newtheorem{proposition}{Proposition}[section]
\newtheorem{remark}{Remark}[section]

\usepackage{bm}

\newtheorem{axiom}{Axiom}
\newcommand{\PN}{{\cal F}(N)}
\newcommand{\PP}{\mbox{\boldmath $\cal R$}}
\newcommand{\ran}{\mbox{\boldmath $le$}}
\newcommand{\dual}{\mbox{\boldmath $d$}}

\newcommand\sm[1]{\textcolor{red}{#1}}

\title{Desirability and social rankings\footnote{Declaration of interests: None.\\  Michele Aleandri is member of the Gruppo Nazionale per l’Analisi Matematica, la Probabilità e le loro Applicazioni (GNAMPA) of the Istituto Nazionale di Alta Matematica (INdAM). \\
Funding: A support of the ANR project THEMIS (ANR-20-CE23-0018) is gratefully acknowledged.}
}

\author{Michele Aleandri\footnote{LUISS University, Roma, Italy; maleandri@luiss.it}, Felix Fritz\footnote{ LAMSADE, CNRS, Université Paris-Dauphine, Université PSL, 75016, Paris, France; felix.fritz@dauphine.psl.eu}, Stefano Moretti\footnote{ LAMSADE, CNRS, Université Paris-Dauphine, Université PSL, 75016, Paris, France; stefano.moretti@dauphine.fr}}
\date{\today}

\begin{document}

\setlength{\abovedisplayskip}{0pt}

\maketitle

\begin{abstract}

In coalitional games, a player $i$ is regarded as {\it strictly more desirable} than player $j$ if substituting $j$ with $i$ within any coalition leads to a strict augmentation in the value of certain coalitions, while preserving the value of the others.
This concept has sparked extensive research, 
particularly within the literature on simple games, with the objective of establishing conditions for the completeness of the desirability relation and ensuring its alignment with diverse power indices.
In this paper, we adopt a property-driven approach to ``integrate'' the notion of desirability relation into a total relation by establishing sets
of independent axioms leading to the characterization of solution concepts from the related literature. 

The solutions considered in this paper extend beyond simple games, accommodating situations where coalitions are partitioned into more than two classes. So, our problem boils down to the challenge of ranking players by considering the ranking of their coalitions, commonly referred to as the {\it social ranking problem}. In this context, a social ranking solution is defined as a map that associates any ranking over the subsets of a finite set $N$ to a total binary relation over the elements of $N$.

We focus on
social ranking solutions consistent with the desirability relation and we propose complementary sets of properties for the axiomatic characterization of five existing solutions: 
the {\it Ceteris Paribus (CP-)majority}, the {\it lexicographic excellence} (lex-cel), the \emph{dual-lex}, the $L^{(1)}$\emph{ solution} and its dual version $L_*^{(1)}$. These novel axiomatic characterizations 
reveal additional similarities among the five solutions and emphasize the essential characteristics that should be taken into account when selecting the most suitable social ranking for a given context.
As an example of application, a practical scenario involving a bicameral legislature is studied. 

\end{abstract}

\noindent\textbf{Keywords}: Social ranking, Axioms, Desirability, Multi-valued simple games, Coalitional ranking

\section{Introduction}\label{sec:intro}

The desirability relation in a simple game entails a (partial) binary relation over the set of players based on their possibilities to form winning coalitions
\cite{I58,taylor2021simple,carreras1996complete}. In words, a player $i$ is at least as desirable as another player $j$ if any winning coalition  containing $j$ remains winning when $j$ is replaced by $i$. Such notion has been generalized for {\it Transferable Utility} (TU)-games in \cite{maschler1966characterization}, where a player $i$ is said to be at least as
{\it desirable} as player $j$, if the worth of any coalition $S$ containing $i$ is at least as large as the worth of the coalition  obtained by  swapping $i$ with $j$ in $S$. Additionally, if the worth of at least one coalition $S$ increases when $j$ is swapped for $i$, then $i$ is said to be {\it strictly} more desirable than $j$. It is not hard to see that a (strict) desirable relation yields a preorder (reflexive and transitive) over the set of players. Many studies have thus focused on understanding under which conditions the desirability relation is complete \cite{taylor1993weighted,taylor2021simple,freixas2008dimension} and how one can guarantee that various {\it power indices} align with it \cite{diffo2002ordinal,carreras2008ordinal,courtin2015note} with the goal  to provide a meaningful ranking of players. Within this research stream, we axiomatically analyse methods aimed to construct total binary relations that are compatible  with the desirability relation and, in addition, satisfy other sets of
independent axioms. 


The methods studied in this paper attempt to rank players according to their influence in {\it multi-valued simple games}, which extend the notion of simple game (where coalitions may be winning or losing) to situations where coalitions are partitioned into more than two classes.
Examples of multi-valued coalitional games from the literature include {\it three-valued simple games} \cite{musegaas2018three} and {\it $(2,k)$ simple voting games} \cite{freixas2005banzhaf,freixas2005shapley}.
As for simple games, where classical $0$ and $1$ values are used to represent losing and winning coalitions, respectively, values associated to coalitions in a multi-valued simple game do not immediately convey a connotation in terms of transferable utility. Nevertheless, their higher magnitudes allow for representing most preferred coalitional outcomes.
For instance, in the context of a multicameral electoral system \cite{taylor1993weighted}, where the approval of an alternative such as a bill or an amendment supported by a coalition needs a sequential approval in a series of elections within distinct chambers, a coalition prevailing in a greater number of chambers might be preferred to one prevailing in a lower one, and so the respective coalitonal values should be chosen according to such preferences. In this context, values are a mere numerical representation of ordinal preferences over coalitional outcomes, so, our framework boils down to the analysis of methods aimed at ranking players based on an ordering of their coalitions. This problem is also known in literature as the {\it social ranking problem} \cite{bernardi:hal-02191137, haret:hal-02103421}. Specifically, a {\it social ranking solution} is defined as a map that associates to any {\it coalitional ranking}, i.e. a ranking over the subsets of a finite set $N$, a total binary relation over the elements of $N$. Several social rankings solutions have been proposed and studied in the literature (\cite{bernardi:hal-02191137,haret:hal-02103421,algaba:hal-03388789,suzuki2024set,suzuki2024consistent,beal2023core}; see also \cite{fritz2023social} for a comprehensive survey).

One of the main goal of this paper is to single out social ranking solutions that align with the strict desirability relation, and to propose complementary sets of properties that, together with strict desirability, can be used to axiomatically characterize those solutions. Precisely, we provide novel axiomatic characterizations of three social ranking solutions from the literature: 1) the CP-majority \cite{haret:hal-02103421} which ranks two players $i$ and $j$ according to a majority rule where ``voters'' are all coalitions $S$ that do not contain neither $i$ nor $j$, and where each  voter-coalition $S$ casts a ballot in favor of the player that yields more, between $S \cup \{i\}$ and $S \cup \{j\}$. It is worth noting that, as a majority rule, the CP-majority, unlike the other two solutions, does not necessarily yield a transitive relation over players; 2) the lexicographic excellence (in short, lex-cel) \cite{bernardi:hal-02191137}, which compares players  in a lexicographic way based on their occurrences in the equivalence classes of a coalitional ranking, starting from the first equivalence class and prioritizing the element occurring more frequently (in case of a tie, it proceeds to the next class, continuing until a difference is found or the last equivalence class is reached with all ties); 3) the $L^{(1)}$ solution \cite{algaba:hal-03388789}, that similar to lex-cel ranks players rewarding their excellence through the comparison of their occurrences within the equivalence classes (from the first to the last one in a lexicographic way), and considering the size of coalitions as a secondary principle according to which small coalitions in the same equivalence class are more likely to form and should count more. Moreover,  by making a simple adjustment to just one axiom, we achieve other two sets of independent axioms that provide new characterizations of a dual version of lex-cel and $L^{(1)}$ named in the literature as dual-lex \cite{bernardi:hal-02191137, suzuki2024consistent} and $L_*^{(1)}$ \cite{algaba:hal-03388789}, respectively.

Besides the alignment with the desirability relation, the sets of axioms used to characterize the five social rankings listed above, share other similarities. First, each set of axioms contains a standard property of {\it neutrality} or {\it symmetry} (the first implying the second, but not the way around), which state that the  players' identity should not affect their ranking (neutrality), or that two players that can be swapped in any coalition without affecting the coalitional ranking, should be ranked equally (symmetry). Another property (namely, {\it consistency after indifference}), specifying that a social ranking should behave consistently across two coalitional rankings when the same set of ties are resolved, is satisfied by all social rankings, but it is only  used in the axiomatic characterization of the CP-majority.

Second, each set of axioms is also characterized by properties that specify, to some extent, the fact that the social ranking of two elements remains {\it invariant} with respect to the identity and, occasionally, the number of other elements within the coalitions they form.
To be more specific, these  axioms are named, respectively, {\it equality of coalitions}, {\it coalitional anonymity} and {\it per-size coalitional anonymity}, with the latter accounting for the cardinality of coalitions.

Finally each set of axioms is characterized by one or two properties of {\it independence} or {\it consistency} for social rankings, which encapsulate some distinctive features of the solutions under study, and can be used to guide the selection of the most suitable ranking based on some specific criteria (namely, these axioms are {\it independence from the worst set}, {\it consistency after indifference} and {\it consistency after tie-breaks}).

Another property, inspired by a  notion of desirability applying to coalitions of the same size, is also used to complete the axiomatic characterization of the $L^{(1)}$ solution and its dual $L_*^{(1)}$, which highlights their peculiarity in considering the cardinality of coalitions in determining the relation between players on specific situations.

The road-map of the paper is as follows. We start in the next section with some preliminaries and the definition of the main notions of social ranking solutions studied in this paper. In Section \ref{sec:axioms}, we proceed by introducing the axioms employed in this paper, accompanied by a comprehensive discussion and interpretation of each.
Sections \ref{sec:main} is aimed to present the main results of this study, that is the axiomatic characterization of the CP-majority (Section \ref{sec:cp}), the lex-cel and its dual
(Section \ref{sec:lexcel}) and the $L^{(1)}$ solution and its dual (Section \ref{sec:lex1}). The logical dependencies among the axioms considered in this paper are presented and discussed in Section \ref{sec:logic}. In Section \ref{sec:cases} we demonstrate how social rankings can be employed to rank political parties in a practical scenario involving a bicameral legislature, where legislators are divided into distinct lower and upper houses, and the approval of a bill necessitates endorsement from both legislative bodies. Section \ref{sec:final} concludes.

	\section{Notations and preliminary notions}

	In the following, $N$ is a finite set  containing $n$ elements, $2^N$ is another finite set whose elements are the subsets of $N$, while $\PN = 2^N \setminus \{\emptyset\}$ denotes the set of non-empty subsets of $N$. 
	Given a set $S \in 2^N$, its cardinality is indicated by $|S|$ (with $|\emptyset|=0$). For notational convinience to write a subset of $N$ we omit brases and commas, i.e. instead of $\{i,j\}\subset N$ we write $ij\subset N$.
A binary relation on $N$ is a set
$R \subseteq N \times N$ and the fact that $(i,j) \in R$, for some $i,j \in N$, is also denoted as $i R j$.
A binary relation $R$ is said \emph{transitive} if  $i R j$ and $j R k$
	imply $i R k$, for
	each $i, j, k \in N$; $R$ is said
 \emph{total} if   $i R j$ or
	$j R i$, for
	each $i, j \in N$; $R$ is said \emph{antisymmetric}, if $i R j$ and $j R i$ imply $i =
	j$, for each $i, j \in N$.	
A transitive and total binary relation on a finite set $N$ is called a \textit{total preorder} or a \textit{ranking} on $N$. A ranking on the set of non-empty subsets of $N$ is called \textit{coalitional ranking}.
	A transitive, total and antisymmetric binary relation on $N$ is called a \textit{total order}. $\bm{\mathcal{B}}(N)$ and $\PP(N)$ denote, respectively, the sets of total binary relations and the sets of rankings (or total preorders) on a finite set $N$.

	To represent a {\it coalitional ranking}, in general we use the notation $\succsim$, with $\succsim \in \bm{\mathcal{R}}(\PN)$. In the related literature, a coalitional ranking $\succsim$ is also called {\it power relation} \cite{haret:hal-02103421, allouche:hal-02930241}. The symmetric part of $\succsim$ is denoted by $\thicksim$ ($S\thicksim T$ if both $S\succsim T$ and $T\succsim S$). Conversely, $\succ$ denotes the asymmetric part ($S\succ T$ if $S\succsim T$ but not $T\succsim S$). Thus, for any pair of subsets $S, T \in \PN$, $S \succ T$ conveys that $S$ is strictly {\it stronger} (or {\it preferred}) than $T$, while $S \thicksim T$ indicates {\it equivalence} in strength (or {\it indifference}) between $S$ and $T$.\\
	Suppose a coalitionl ranking $ \succsim \in  \PP(\PN)$ is of the form\footnote{As usual, in the remaining of the paper, if a ranking is represented using a chain form like $i R j R k R l R...$ one must infer by transitivity the relation between any pair of elements that are not adjacent in the chain. So, for instance, by transitivity of $R$ we have that $i R k$.}
	$$S_1 \succsim  S_2  \succsim S_3  \succsim \dots \succsim  S_{2^n-1}.$$
	Given this ranking $ \succsim$, consider its \textit{quotient order}, denoted as follows
	$$\Sigma_1 \succ  \Sigma_2  \succ  \Sigma_3  \succ  \dots \succ  \Sigma_{l}$$ in which the subsets $S_j$ are grouped in  the  \textit{equivalence classes} $\Sigma_k$ generated by the symmetric part of $\succsim$.
	This means that all the sets in $\Sigma_1$ are indifferent to $S_1$ and are strictly better than the sets in $\Sigma_2$  and so on.
	Observe that    $\Sigma_i=\{S_i\}$ for any $i=1,\dots,2^{n}-1$ if and only if  $\succsim$ is a total order.
	
	\begin{example}\label{ex2.1}
		As an example of power relation in a practical situation, consider a bicameral legislature where some political parties of a set $N$ are represented in two distinct houses, namely,  \emph{Lower House} and \emph{Upper House}, by $n_{\ell}$ and $n_u$ members, respectively. Precisely, suppose that each political party $i$ in $N$ is represented by $w^{\ell}_i$ members in the Lower House and $w^u_i$ in the Upper House, and that the Lower House has the authority to either accept (or reject) a bill based on whether there is (or is not) a majority of at least $q^\ell$ of its members in favour. If the Lower House approves the bill, it is then transmitted to Upper House, where a similar process occurs, requiring a majority of at least $q^u$ members for acceptance. The bill becomes law only if it gains approval from both houses. Notably, only the Lower House holds the prerogative to initiate the submission of a bill and commence the entire procedure.
		As it was initially suggested and discussed in \cite{musegaas2018three}, such a situation can be represented as a three-valued simple game $(N,v)$  where the function $v:2^N \rightarrow \{0,1,2\}$ is such that
		\begin{equation}\label{3valuesg}
			v(S)=\left\{
			\begin{array}{ll}
				2, & \mbox{if } \sum_{i \in S}w^\ell_i \geq q^\ell \mbox{ and } \sum_{i \in S}w^u_i \geq q^u; \\
				1, & \mbox{if } \sum_{i \in S}w^\ell_i \geq q^\ell \mbox{ and } \sum_{i \in S}w^u_i < q^u;    \\
				0, & \mbox{otherwise.}
			\end{array}
			\right.
		\end{equation}
\end{example}
		In the scenario introduced in Example \ref{ex2.1}, labels linked to coalitions in relation \eqref{3valuesg} are purely ordinal. The ordering of $2>1>0$ merely signifies a rising preference of political parties regarding  outcomes that coalitions may guarantee, without quantifying a specific level of utility.
		Due to its ordinal nature, such a situation can be described as a coalitional ranking $\succsim \in  \bm{\mathcal{R}}(\PN)$   such that   for any  coalitions $S,T \in \PN$:
		\begin{equation}\label{3valuesgrank}
			S \succsim T \Leftrightarrow v(S) \geq v(T).
		\end{equation}

	\begin{example}\label{ex3.1bis}
		Consider a bicameral legislation with two houses, as in Example \ref{ex2.1}, which are composed by members of three political parties in $N=\{1,2,3\}$ having weights $(w^\ell_1, w^\ell_2, w^\ell_3)=(2,2,2)$,  $(w^u_1, w^u_2, w^u_3)=(2,1,2)$ and the same quota in the two houses $q^1=q^2=4$.
		The corresponding three-valued game according to relation (\ref{3valuesg})  gives value 2 to coalitions $\{1,2,3\}$ and $\{1,3\}$, value 1 to coalitions $\{1,2\}$ and $\{2,3\}$ and value 0 to all singleton coalitions.
		According to (\ref{3valuesgrank}), this scenario can be represented as a coalitional ranking $\succsim \in \bm{\mathcal{R}}(\PN)$ such that
		\[
		123 \sim 13 \succ 12 \sim 23 \succ 1 \sim 2 \sim 3.
		\]
	\end{example}

	Given a coalitional ranking over the subsets of $N$, there is a compelling interest in attempting to compare the relative importance of the single elements within $N$ taking into account the position of  coalitions they may form.
	For instance,  comparing the ``{\it relevance}'' of political parties in a bicameral legislature, as outlined in Example \ref{ex2.1}, can be done by looking at their potential possibilities to form coalitions with high levels of $v$ or, equivalently, with high positions in $\succsim$. A social ranking solution\footnote{Notice that in the literature, the co-domain of a social ranking has been often defined as the set of total preorders on $N$ (see, for instance, the seminal papers \cite{bernardi:hal-02191137,algaba:hal-03388789}. In this paper, we consider a larger co-domain in order to make a consistent analysis of the lex-cel and the $L^{(1)}$ solution (and their dual versions) with the CP-majority, that in general does not yield a transitive relation for $|N|>2$.} seeks to address this problem by linking to every coalitional ranking a total binary relation on the set $N$ aimed at comparing the ordinal relevance of each pair of elements in $N$.
	
	\begin{definition}
		\em 
		A \emph{social ranking (solution)}, or briefly, a \emph{solution},  is a function $$R:\bm{\mathcal{R}}(\PN)\to \bm{\mathcal{B}}(N).  $$
	\end{definition}
	The  image of a coalitional ranking $\succsim \in \bm{\mathcal{R}}(\PN)$ under a social ranking $R$   is a total binary relation on $N$, notated as $R^\succsim$.
	 The semantics of a relation between two elements $i$ and $j$ expressed as $i R^{\succsim} j$ is described as follows: ``$i$ is at least as relevant as  $j$, when the social ranking $R$ is applied to the coalitional ranking $\succsim$''. We also denote by $I^\succsim$ its symmetric part ($i I^\succsim j$ if $(i,j)\in R^\succsim$ and $(j,i)\in R^\succsim$) and, conversely, by $P^\succsim$ its asymmetric part ($i P^\succsim j$ if $(i,j)\in R^\succsim$ but $(j,i)\notin R^\succsim$).

As discussed in Section \ref{sec:intro}, several social rankings have been introduced in the literature. In this paper we will focus on three particular ones: the  CP-majority \cite{haret:hal-02103421}, the  lex-cel \cite{bernardi:hal-02191137}, the  $L^{(1)}$ solution  \cite{algaba:hal-03388789}.  In the remaining of this section, we recall the formal definition of each of those solutions, and we illustrate their computation on short numerical examples. The dual versions of the lex-cel and of $L^{(1)}$ will be also introduced and studied in Sections \ref{sec:lexcel} and \ref{sec:lex1}, respectively.\\ 

\noindent
The CP-majority compares two elements $i$ and $j$ in $N$ by uniquely considering the ranking of pairs of subsets $S \cup \{i\}$ and $S \cup \{j\}$, where $S$ is a coalition excluding both $i$ and $j$. This kind of comparisons, also named {\it CP-comparisons} \cite{haret:hal-02103421}, involve coalitions that differ only in the exclusion of either $i$ or $j$, while everything else remains the same.
		Before  recalling the formal definition of  the  {\it Ceteris Paribus} relation (\emph{(CP)-majority}) \cite{haret:hal-02103421}, we introduce some further notations. Given a coalitional ranking $\succsim \in \bm{\mathcal{R}}\left(\PN\right)$ and two elements $i, j \in N$ we define two sets: the set  $D_{i j}(\succsim)=\{S \in$ $\left.2^{N \backslash\{i, j\}}: S \cup\{i\} \succ S \cup\{j\}\right\}$ of coalitions $S$ (not containing neither $i$ nor $j$) such that $S \cup \{i\}$ is strictly stronger than $S\cup \{j\}$, and the set $E_{ij}(\succsim)$ $=\{S \in$ $\left.2^{N \backslash\{i, j\}}: S \cup\{i\} \sim S \cup\{j\}\right\}$ of coalitions $S$ (not containing neither $i$ nor $j$) such that $S \cup \{i\}$ and $S\cup \{j\}$ are equally strong. So, the cardinality of set $D_{i j}(\succsim)$ ($D_{j i}(\succsim)$) represents the number of times $i$ ($j$) defeats $j$ ($i$), and the definition of the CP-majority relation between $i$ and $j$ in $\succsim$ follows according to the comparison between $|D_{i j}(\succsim)|$ and $|D_{j i}(\succsim)|$.
		
		\begin{definition}(CP-majority \cite{haret:hal-02103421}.) The CP-majority is the solution $R_{CP}:\PP(\PN)\to \bm{\mathcal{B}}(N)$ such that for all $\succsim \in
			\bm{\mathcal{R}}(\PN)$ and $i, j \in N$ :
			$$
			i R_{CP}^{\succsim} j \Leftrightarrow |D_{i j}(\succsim)| \geq |D_{j i}(\succsim)|.
			$$
			(Following the same convention as before, $I_{CP}^\succsim$ and $P_{CP}^\succsim$ stand for the symmetric part and the asymmetric part of $R_{CP}^\succsim$, respectively.)
		\end{definition}

	\begin{example}\label{ex3.1CP}
		Consider the coalitional ranking of Example \ref{ex3.1bis}. For instance, we have that $D_{1 2}(\succsim)=\{3\}$ whereas $D_{2 1}(\succsim)=\emptyset$. Since $|D_{1 2}(\succsim)|>|D_{2 1}(\succsim)|$, it follows that the CP-majority ranks party $1$ strictly more relevant than $2$, that is notated as $1 P_{C P}^{\succsim} 2$.
We have the following matrix $(|D_{i j}(\succsim)|)_{i,j \in N}$
$$
(|D_{i j}(\succsim)|)_{i,j \in N}=\left ( \begin{array}{lll}
         0 & 1 & 0 \\
         0 & 0 & 0 \\
         0 & 1 & 0
         \end{array} 
         \right )
$$ 
which yields, in this case,  a ranking such that $1 I_{C P}^{\succsim} 3 P_{C P}^{\succsim} 2$, i.e. parties $1$ and $3$ are equally relevant and both of them are strictly more relevant than party $2$ according to the CP-majority on the coalitional ranking $\succsim$.
	\end{example}
\noindent
The lex-cel solution \cite{bernardi:hal-02191137} is based on the idea that individuals who consistently appear at the top of a coalitional ranking should be ranked as the most relevant. It compares any pair of elements by examining their occurrences in equivalence classes of a coalitional ranking, starting with the class with the lowest index (i.e., the class of top coalitions). If there is a tie, it proceeds to compare occurrences in subsequent classes until a difference is found or all classes show ties, in which case the elements are declared indifferent in terms of their relevance.

We first need some more  notations. Given a coalitional ranking $\succsim\in \mathcal{R}(\PN)$ with quotient order $\Sigma_1 \succ  \Sigma_2  \succ  \Sigma_3  \succ  \dots \succ  \Sigma_{l}$, denote by $i_k$ the number of sets in $\Sigma_{k}$ containing $i$:	\[ i_k = |\{S\in \Sigma_k: i\in S\}| \]
		for $k=1, \dots,l$. Now, let $\theta^\succsim(i)$ be the $l$-dimensional vector $\theta^\succsim(i)=(i_1,\dots,i_l)$ associated to $\succsim$. Consider the lexicographic order among vectors:
		$$\mathbf{i} \ge_L \mathbf{j} \quad  \mbox{ if either } \mathbf{i}=\mathbf{j} \;\;
		\mbox{ or }\;  \exists s: i_k=j_k,\ k=1,\dots, s-1\; \mbox{ and } i_s>j_s.$$
		 The lex-cel is defined as follows.

		\begin{definition}[lex-cel \cite{bernardi:hal-02191137}]
			The \emph{lexicographic excellence} (lex-cel)  solution is the function $R_{\ran}:\PP(\PN)\to \bm{\mathcal{B}}(N)$ defined for any coalitional ranking $\succsim \in \PP(\PN)$ as
			$$
			i R_{\ran}^\succsim j \qquad {\rm if\;} \qquad \theta^\succsim (i)\;\;\ge_L \;\;\theta^\succsim (j).
			$$
			(Following the same convention as before, $I_{\ran}^\succsim$ and $P_{\ran}^\succsim$ stand for the symmetric part and the asymmetric part of $R_{\ran}^\succsim$, respectively.)
		\end{definition}

 We illustrate its computation in the following example.

\begin{example}\label{ex3.1lexcel}
Consider again the coalitional ranking of Example \ref{ex3.1bis}. Its quotient order is
		\[
		\Sigma_1 \succ  \Sigma_2  \succ  \Sigma_3
		\]
		with $\Sigma_1=\{ 123, 13\}$, $\Sigma_2=\{ 12, 23\}$ and $\Sigma_3=\{ 1, 2, 3\}$. Then we have that 
		\[
		\begin{array}{l}
		\theta^\succsim(1)=(2,1,1),\\
		\theta^\succsim(2)=(1,2,1),\\
		\theta^\succsim(3)=(2,1,1).
		\end{array}
		\] 
		It follows that the lex-cel yields the ranking   $1 I_{\ran}^{\succsim} 3 P_{\ran}^{\succsim} 2$, i.e element $1$ and $3$ are ranked with equal relevance and each of them  is strictly more relevant than element $2$ according to the lex-cel on the coalitional ranking $\succsim$.
	\end{example}

\noindent
We finally introduce the third main social ranking considered in this article, the $L^{(1)}$ solution introduced in \cite{algaba:hal-03388789}. Similar to lex-cel, the $L^{(1)}$ solution also seeks to recognize the contributions of elements to top coalitions within a calitional ranking. However, in the lexicographic comparisons of their occurrences in the equivalences classes of a coalitional ranking, from the top one to the worst one, it accounts for the size of coalitions, thus emphasizing the significance of smaller ones.

		Given a coalitional ranking $\succsim \in \PP(\PN)$ and an agent $i \in N$, define the matrix $M^{\succsim, i}$ of size $(n, l)$ where each entry $M_{s k}^{\succsim, i}$ denotes the number of coalitions of size $s \in\{1, \ldots, n\}$ containing $i$ in the equivalence class $\Sigma_k, k \in\{1, \ldots, l\}$ of $\succsim$. 
\begin{definition}[$L^{(1)}$  \cite{algaba:hal-03388789}] \label{l1sol}
			The $L^{(1)}$ solution is the map $R_{L^{(1)}}:\PP(\PN)\to \bm{\mathcal{B}}(N)$
			such that for any  coalitional ranking $\succsim \in \PP(\PN)$ and $i,j \in N$:
			$$
			i \ I_{L^{(1)}}^\succsim \ j \Longleftrightarrow M^{\succsim, i}=M^{\succsim, j}
			$$
			and
			$$
			i \ P_{L^{(1)}}^\succsim \ j \Longleftrightarrow
			\begin{array}{l}
				\exists \hat{s} \in \{1, \ldots, n\} \mbox {and } \hat{k} \in  \{1, \dots, l-1\} \mbox{ such that } \\
			\qquad \left\{\begin{array}{ll}
				M^{\succsim, i}_{sk} = M^{\succsim, j}_{sk},                         & \forall s \in \{1, \ldots, n\}, \forall k \in \{1, \ldots, \hat{k}-1\}; \smallskip \\
				M^{\succsim, i}_{s\hat{k}} = M^{\succsim, j}_{s\hat{k}},             & \forall s \in \{1, \ldots, \hat{s}-1\}; \smallskip                                 \\
				M^{\succsim, i}_{\hat{s}\hat{k}} > M^{\succsim, j}_{\hat{s}\hat{k}}. &
			\end{array}\right.
		\end{array}
		$$
					(Following the same convention as before, $I_{L^{(1)}}^\succsim$ and $P_{L^{(1)}}^\succsim$ stand for the symmetric part and the asymmetric part of $R_{L^{(1)}}^\succsim$, respectively.)

\end{definition}

	\begin{example}\label{ex3.1l1}
		Again consider the coalitional ranking of Example \ref{ex3.1bis}. 
We have that 
\[
M^{\succsim, 1}=\left(\begin{array}{lll} 0 &0&1\\ 1 &1&0\\ 1&0&0  \end{array}\right), \ \ 
M^{\succsim, 2}=\left(\begin{array}{lll} 0 &0&1\\ 0 &2&0\\ 1&0&0  \end{array}\right), \ \ 
M^{\succsim, 3}=\left(\begin{array}{lll} 0 &0&1\\ 1 &1&0\\ 1&0&0  \end{array}\right).
\]

		It follows that the $L^{(1)}$ solution yields the ranking   $1 I_{L^{(1)}}^{\succsim} 3 P_{L^{(1)}}^{\succsim} 2$, i.e. parties $1$ and $3$ are ranked with equal relevance (for $M^{\succsim, 1}$ and $M^{\succsim, 3}$ are identical), and each of them  is strictly more relevant than party $2$,  according to Definition \ref{l1sol} with $\hat{s}=2$ and $\hat{k}=1$.
	\end{example}
	
Section \ref{sec:main} will also provide limited examination of dual versions of both lex-cel and the $L^{(1)}$ solution.

	\section{Properties for social rankings}\label{sec:axioms}

	In this section, we introduce and discuss various properties for social ranking solutions, and we explore their implications within specific contexts. We start with the main property studied in this work.

	\begin{axiom}[Strict Desirability (SDes)]\label{Des2}
		Let $i,j \in N$. A solution $R: \bm{\mathcal{R}}(\PN) \longrightarrow \bm{\mathcal{B}}(N)$ satisfies the  \emph{Strict Desirability}  property  if
		$$i \ P^\succsim \ j $$
		for any  $\succsim\in\bm{\mathcal{R}}(\PN)$ such that $S \cup \{i\} \succsim S \cup \{j\}$ for all $S \in 2^{N\setminus \{i,j\}}$ and there exists a coalition $T \in 2^{N\setminus \{i,j\}}$ such that $T \cup \{i\} \succ T \cup \{j\}$.
	\end{axiom}
	The concept of strict desirability plays an important role for various solution concepts in cooperative game theory \cite{carreras2008ordinal,holler2013power}.
 The ability of a social ranking to discriminate the higher relevance of an individual $i$ against another individual $j$ based on the fact that for some  coalition $T$ not containing neither $i$ nor $j$ we have that   $T \cup \{i\}$ is (strictly) stronger than $T \cup \{i\}$ (while it is never true that $S \cup \{j\}$ is strictly stronger than $S \cup  \{i\}$ for any coalition $S$) mirrors the responsiveness of the solution in assessing the relative contributions of players across the entire spectrum of coalitions.
	
	\begin{example}\label{ex3.1bisnew}
		Consider the coalitional ranking of Example \ref{ex3.1bis}.
		A social ranking $R$ satisfying strict desirability ranks party $3$ and $1$ strictly better than $2$ (i.e., $1 P^\succsim 2$ and  $ 3 P^\succsim 2$) thanks to the fact that parties $1$ and $3$ together form the only coalition having a majority in both houses (so, $13 \succ 12$ and $13\succ 23$), while singleton coalitions are all indifferent, making both $1$ and $3$ strictly ``more desirable'' than $2$. Instead, the property of strict desirability does not specify any ranking between the two symmetric parties $1$ and $3$.
	\end{example}
	
	According to the strict desirability property, a relevance relation between an individual $i$ over another individual $j$ can be dictated by even a single coalition $S \cup \{i\}$ that is strictly stronger than $S \cup \{j\}$, if all other comparisons of the form $S \cup \{i\}$ {\it versus} $S \cup \{j\}$, with $S$ belonging to the set of all possible coalitions excluding $i$ and $j$, 
 state indifference. This type of comparisons $S \cup \{i\}$ {\it versus} $S \cup \{j\}$ are also called CP-comparisons \cite{haret:hal-02103421}. One could argue that the reliability of some CP-comparisons may be questioned under certain circumstances; for instance, if the involved coalitions are deemed implausible or if their probability to form is considered not sufficiently high, due to a political contrast between some of their members. 
	Nevertheless, throughout this article, we shall operate under the assumption that all coalitions are feasible, thereby adopting the strict desirability property as a pivotal axiom for all social ranking solutions aiming to capture any evidence of  coalitions' optimality.
	
	The following axiom, symmetry, is another classical property for several game theoretical solutions. It states that if two elements $i$ and $j$ are interchangeable within any coalition $S$ without altering the positions of coalitions in the coalitional ranking (in other terms, $i$ and $j$ are perfect substitutes over all possible coalitions they may form), then they should share the same position also in the social ranking. For instance, in Example \ref{ex3.1bis}, parties $1$ and $3$ should be ranked indifferent relevant according to a symmetric social ranking.
	\begin{axiom}[Symmetry (Sym)]\label{Sym}
		Let $i,j \in N$. A solution $R: \bm{\mathcal{R}}(\PN) \longrightarrow \bm{\mathcal{B}}(N)$ satisfies the  \emph{Symmetry}  property  if
		$$i \ I^\succsim \ j $$
		for any  $\succsim\in\bm{\mathcal{R}}(\PN)$ such that $S \cup \{i\} \sim S \cup \{j\}$ for all $S \in 2^N\setminus \{i,j\}$.
	\end{axiom}

	The next property of neutrality  for social rankings has been introduced in \cite{haret:hal-02103421} to ensure that a social ranking treats two elements regardless of their individual identities.
	
	
	\begin{axiom}[Neutrality (Neu)]\label{N}
		A solution $R: \bm{\mathcal{R}}(\PN) \longrightarrow \bm{\mathcal{B}}(N)$ satisfies the property of Neutrality if and only if for all power relation $\succsim,\succsim'\in \bm{\mathcal{R}}(\PN)$ and  $i, j \in N$, such that
		\begin{equation*}
			S \cup\{i\} \succsim S \cup\{j\} \Leftrightarrow S \cup\{j\} \succsim' S \cup\{i\}
		\end{equation*}
		for all $S \in 2^{N \backslash\{i,j\}}$, it holds that
		\begin{equation*}
			i R^{\succsim} j \Leftrightarrow j R^{\succsim'} i.
		\end{equation*}
	\end{axiom}
	In an adjusted scenario from Example \ref{ex3.1bis} 
	imagine transferring one unit of weight from party $3$ to party $2$ in the Upper House, altering the weights to $(2,2,1)$. This shift grants the $12$ coalition a majority in both Houses, unlike the $13$ coalition, which can only secure a majority in the Lower House. Consequently, party $2$ essentially swaps its position with party $3$ in the coalitional ranking. Given that only the legislative weight of parties should influence their ranking, but not, i.e., their names or other irrelevant attributes, a social ranking satisfying neutrality dictates that, before the transfer,  party $3$ is at least as relevant as party $2$ if and only if  party $2$ is at least as relevant as  party $3$, after the transfer.
	Notice that the neutrality property also implies the symmetry one. Thus, a social ranking that satisfies neutrality also ranks parties 1 and 3 equally relevant in Example \ref{ex3.1bis}  (see Proposition \ref{logic1} for more details).

	A further property introduced in \cite{haret:hal-02103421} with the name of Equality of Coalitions, requires that for any pair of individuals $i$ and $j$, all CP-comparisons $S \cup \{i\}$ {\it vs.}  $S \cup \{j\}$, for any $S$ containing neither $i$ nor $j$, matters equally in establishing the social ranking between individuals $i$ and $j$, regardless of the size or structure of each coalition $S$ involved in the CP-comparison.

	\begin{axiom}[Equality of Coalitions (EC)]\label{EC} A solution $R: \bm{\mathcal{R}}(\PN) \longrightarrow \bm{\mathcal{B}}(N)$ satisfies the property of \emph{Equality of Coalitions} if and only if for all power relations $\succsim, \succsim' \in$ $\bm{\mathcal{R}}(\PN)$, $i, j \in N$, and a bijection $\pi$ on $2^{N\setminus\{i, j\}}$ such that
		$$ S \cup\{i\} \succsim S \cup\{j\} \Leftrightarrow \pi(S) \cup\{i\} \succsim' \pi(S) \cup\{j\}$$
		for all $S \in 2^{N \backslash\{i,j\}}$, it holds that
		$$i R^{\succsim} j \Leftrightarrow i R^{\succsim'} j.$$
	\end{axiom}
	Notice that taking a bijection $\pi$  such that $\pi(S)=S$ for all $S \in 2^{N \setminus \{i,j\}}$,
	Axioms \ref{EC} states that a social ranking solution should remain unaffected by coalitional comparisons other than the {\it ceteris paribus} ones, i.e. those involving sets of the form $S \cup \{i\}$ and $S \cup \{j\}$. Differently stated, Equality of Coalitions requires that the relation specified by a social ranking between elements $i$ and $j$ should be independent of all comparisons that are not CP for $i$ and $j$, as it is also illustrated in the following example.

	\begin{example}\label{exEC}
		Consider the coalitional ranking $\succsim$ from Example \ref{ex3.1bis} 
		and a new coalitional ranking $\succsim'$ that may emerge as a result of alterations in the weights of the parties in the two houses 
		(the specific weight details are omitted for brevity) such that
		\[
		123 \sim' 12 \succ' 13 \sim' 23\sim' 1 \succ'2 \sim' 3.
		\]
		Notice that in $\succsim$ we have $13 \succ 23$ and $1 \sim 2$, while in $\succsim'$ we have that
		$13 \sim' 23$ and $1 \succ' 2$. So, defining a bijection    $\pi$ on $\{\emptyset,3\}$ such that $\pi(\emptyset)=3$, we have that $S \cup\{1\} \succsim S \cup\{2\} \Leftrightarrow \pi(S) \cup\{1\} \succsim' \pi(S) \cup\{2\}$ for all $S \in \{\emptyset, 3\}$. Then, a social ranking satisfying EC should rank parties $1$ and $2$ in $\succsim$ precisely as in $\succsim'$.
		
		Moreover, the fact that coalition $1$ is in the intermediate equivalence class in the quotient order of $\succsim'$ together with some coalitions of size two, or that $12$ is in the first equivalence, has no effect on the final social ranking between $1$ and $2$. A social ranking satisfying EC must rank $1$ and $2$  independently of the relative position of coalitions that are not involved in CP-comparisons for $1$ and $2$.
	\end{example}
	
	In a political context, the EC property could be of some interest to compare the relevance of political parties with respect to their ability to withstand potential shifts in alliances of the other parties in a coalition, irrespective of their size or their political identities. However, one could argue that forging alliances in smaller coalitions is more likely in a practical political landscape, or that other type of comparisons other than the CP-ones play a role. In this scenario, using EC might not be suitable, and employing alternative versions that relies on other kind of comparisons could prove more effective.


	Another property that is based on the use of a bijection that transforms each coalition (not containing two individuals $i$ and $j$) into another is the coalitional anonymity property introduced in \cite{bernardi:hal-02191137}.
	
	\begin{axiom}[Coalitional Anonymity (CA)]\label{CA}
		Suppose that for two rankings $\succsim, \succsim'  \in \PP(\PN)$, there are two elements $i,j \in N$, and a bijection $\pi$ on $2^{N \setminus\{i,j\}}$ such that, for all $S,T \in 2^{N\setminus\{i,j\}}$:
		\begin{equation}\label{hypAn}
			S \cup \{i\} \succsim T \cup \{j\} \Leftrightarrow \pi(S) \cup \{i\} \succsim' T \cup \{j\}.
		\end{equation}
		Then a solution $R$ satisfies \emph{Coalitional Anonymity}  if the following holds:
		$$i R^\succsim j \Leftrightarrow i R^{\succsim'} j.$$
	\end{axiom}
	
	Unlike the EC property, the CA property does not require that the ranking between $i$ and $j$ is based exclusively on CP-comparisons  $S \cup \{i\}$ {\it vs.} $S \cup \{j\} $, for every $S \in 2^{N\setminus\{i,j\}}$.  The notion of relevance considered by a solution satisfying the CA axiom refers to the position of coalitions containing ether $i$ or $j$ in the coalition ranking, regardless of the size and structure of the coalition to which they belong. Consequently, also the impact of a relation $S \cup \{i\}  \succ S \cup \{j\}$ to a social ranking may differ according to the position of $S \cup \{i\}$  and  $S \cup \{j\}$ in the coalitional ranking.
	Furthermore, CA embodies a principle of position independence of coalitions containing both $i$ and $j$, or neither.
	
	\begin{example}\label{exCA}
		Consider the coalitional ranking $\succsim$ from Example \ref{ex3.1bis} 
		and a new coalitional ranking $\succsim'$ reflecting a shift in  power balance from 3 to 1 that, for some political reason or consensus drift, has reversed the relative position of coalitions $1$ with coalition $13$ as follows
		\[
		123 \sim' 1 \sim' 12 \succ' 23 \succ' 13 \sim' 2 \sim' 3.
		\]
		Defining a bijection    $\pi$ on $\{\emptyset, 3\}$ such that $\pi(\emptyset)=3$, we have that $S \cup\{1\} \succsim T \cup\{2\} \Leftrightarrow \pi(S) \cup\{1\} \succsim' T \cup\{2\}$ for all $S \in \{\emptyset, 3\}$.
		Then, a social ranking satisfying CA ranks parties $1$ and $2$ in $\succsim$ precisely as in $\succsim'$.
		Moreover, the change of position of coalition $12$ that in $\succsim'$ is positioned in the first equivalence class, does not affect the social ranking between $1$ and $2$ (as it would not affect such a ranking a change in the position of coalition $3$ or $123$).

		Notice that in $\succsim'$, $1\succ'2 $ and $23\succ' 13 $, while, as already said in Example \ref{exEC}, in $\succsim$ we have $13 \sim 23$ and $1 \succ 2$. So, there is no way to put in one-to-one correspondence the CP-comparisons $S \cup \{1\}$ {\it vs.} $S \cup \{2\}$ via a bijection that preserves the corresponding coalitional rankings as demanded by Axiom \ref{EC}. So, Axiom \ref{EC} does not specify anything about the relation of social rankings between $1$ and $2$ in $\succsim$ and $\succsim'$.
		In a similar way, one can check that  condition (\ref{hypAn}) demanded by the CA property does not apply when $\succsim$ is considered together with $\succsim'$ of Example \ref{exEC}.
		Differently stated, properties EC and CA may apply to distinct pairs of coalitional ranking. This fact leads to the consequence that is possible to find social ranking solutions that satisfy the Axiom \ref{EC} (EC) but not the Axiom \ref{CA} (CA), and the way around. A formal proof of the logical independence between Axioms \ref{EC} (EC) and \ref{CA} (CA) is presented in Proposition \ref{logic1}.
	\end{example}
	
	In a political framework, Axiom \ref{CA} suggests to rank parties according to the position of coalitions they may form independently from each other, no-matter the identities or the number of the other members in the coalitions. In alternative, one may argue that at least the size of coalitions should matter, and therefore, using a slight different definition as the one used previously, the property of per-size coalitional anonymity, where each coalition is mapped into a coalition of the same cardinality, is recalled from \cite{algaba:hal-03388789}.
	
	\begin{axiom}[Per-size Coalitional Anonymity (PCA)]\label{PCA}
			Suppose that, for 
			$\succsim, \succsim'  \in \PP(\PN)$, there are two elements $i,j \in N$, and a bijection $\pi$ on $2^{N \setminus\{i,j\}}$ such that $|\pi(S)|=|S|$ for all $S \in 2^{N\setminus\{i,j\}}$ and such that
			\begin{equation}\label{hypAn2}
				S \cup \{i\} \succsim T \cup \{j\} \Leftrightarrow \pi(S) \cup \{i\} \succsim' T \cup \{j\}
			\end{equation}
			for all $S,T \in 2^{N\setminus\{i,j\}}$ with $|S|=|T|$.
			
			A solution $R$ satisfies \emph{Per-size Coalitional Anonymity}  if the following holds:
			$$i R^\succsim j \Leftrightarrow i R^{\succsim'} j.$$
		\end{axiom}

		The interpretation of the PCA property is obvious and follows the same lines of CA, except that it applies to a restricted set of pairs of coalitional rankings. For instance, the PCA property does not say anything about the relation between the coalitional ranking $\succsim$ of Example \ref{ex3.1bis} and $\succsim'$ of Example \ref{exCA}, as the unique bijection $\pi$ on $2^{N\setminus\{1,2\}}=\{\emptyset, \{3\}\}$ satisfying the condition $|\pi(S)|=|S|$ is the identity $\pi(S)=S$. It is also clear that a social ranking satisfying Axiom \ref{CA} (CA) also satisfies Axiom \ref{PCA}(PCA) (see Proposition \ref{logic1} and \cite{algaba:hal-03388789} for more details).
		
		The following property addresses scenarios in which two individuals $i$ and $j$ initially hold equal relevance in two distinct coalitional rankings. Subsequently, ties among coalitions are broken, resulting in new (complete) coalitional rankings. The consistency after tiebreaks  property states that the ranking of $i$ relative to $j$ should remain consistent in the revised social rankings.
		\begin{axiom}[Consistency After Tiebreaks (CAT)]\label{CAT}
			A solution $R:\PP(\PN) \to \mathcal{B}(N)$ satisfies the property of \emph{Consistency After Tiebreaks} if  for all power relations $\succsim, \succsim' \in \PP(\PN)$ such that $i I^{\succsim} j$ and $i I^{\succsim'} j$
			we have that
			$$
			i R^{\succsim \setminus B} j \Leftrightarrow i R^{\succsim' \setminus B} j
			$$
			for all $B \subseteq \succsim \cap \succsim'$ such that $(\succsim \setminus B) \in \PP(\PN)$ and $(\succsim' \setminus B) \in \PP(\PN)$.
		\end{axiom}
			\begin{remark} \label{rem:cat}
Notice that the set $B$ of removed pairs of coalitions  considered in the definition of the CAT axiom is well defined. In fact, the requirement that the new relations $(\succsim \setminus B) \in \PP(\PN)$ and $(\succsim' \setminus B) \in \PP(\PN)$ obtained after the removal are still total preorders,  implies that only pairs of coalitions  that are involved in a tie can be removed from $\succsim$ and $\succsim'$ (differently stated, it must be verified that if $(S,T) \in B$, then $\{(S,T),(T,S)\} \subseteq \succsim$ and $\{(S,T),(T,S)\} \subseteq \succsim'$ and, in addition, $(T,S) \notin B$). Moreover, since the removed pairs from $\succsim$ and $\succsim'$ in $B$ must be in the intersection $\succsim \cap \succsim'$, both quotient orders of $\succsim$ and $\succsim'$ must have at least one identical equivalence class $\Sigma$ in their respective quotient orders. In fact, suppose pairs in $B$ are removed from two different equivalence classes. Then, in order to preserve the transitivity of the symmetric relation in $(\succsim \setminus B)$ and $(\succsim' \setminus B)$, the removal of pairs in $B$ would also determine the removal of  different pairs in $\succsim$ and $\succsim'$, which yields a contradiction with the fact that only the same pairs in $B$ can be removed from both $\succsim$ and $\succsim'$. Precisely,  let the quotient orders $\succ$ and $\succ'$ and $\Sigma$ be such that
$$
\Sigma_1 \succ \ldots \hat{\succ} \Sigma \succ  \ldots \succ \Sigma_l,
$$
and
$$
\Sigma'_1 \succ' \ldots \succ' \Sigma \succ'  \ldots \succ' \Sigma'_m.
$$
Denote the new quotient orders $\hat{\succ}=\succ\setminus B$ and $\hat{\succ}'=\succ'\setminus B$, they must be such that
$$
\Sigma_1 \hat{\succ} \ldots \hat{\succ} \Sigma\setminus \Omega\ \hat{\succ}\ \Omega\  \hat{\succ} \ldots \hat{\succ} \Sigma_l,
$$
and
$$
\Sigma'_1 \hat{\succ}' \ldots \hat{\succ}' \Sigma\setminus \Omega \ \hat{\succ}' \ \Omega \  \hat{\succ}' \ldots \hat{\succ}' \Sigma'_m,
$$
where $S \in \Omega$, $T \in \Sigma \setminus \Omega$ and $(S,T) \in B$. 
\end{remark} 

		\begin{example}\label{exCAT}
			
			Consider the coalitional ranking $\succsim$ from Example \ref{ex3.1bis} 
			and a new coalitional ranking $\succsim^{'}\in \PP(\PN)$  such that
			\[
			123 \sim^{'} 13 \sim^{'} 12 \sim^{'} 23 \succ^{'} 1 \sim^{'} 2 \sim^{'} 3 .
			\]
			Suppose that a social ranking $R$ ranks $1$ and $3$ equally relevant in both coalitional rankings, $1 I^{\succsim} 3$ and $1 I^{\succsim^{'}} 3$ (for instance, because $1$ and $3$ are symmetric).
			
			Now, let $B=\{(3,1),(2,1),(3,2)\}$  and suppose that in both rankings $\succsim$ and $\succsim'$ the elements of $B$ are removed so that two new coalitional rankings   $\hat{\succsim}=(\succsim \setminus B)$ and $\hat{\succsim}^{'}=(\succsim' \setminus B)$, both in $\in \PP(\PN)$, are produced such that
			
				\[
			 123   \ \hat{\sim}     \   13   \ \hat{\succ}   \    12 \  \hat{\sim}  \     23  \ \hat{\succ}     \  1  \ \hat{\succ}    \    2  \ \hat{\succ}   \     3,
				\]
				and
				\[
				123  \ \hat{\sim}^{'} \ 13  \ \hat{\sim}^{'} \ 12 \ \hat{\sim}^{'} \  23 \ \hat{\succ}^{'} \ 1  \ \hat{\succ}^{'} \ 2 \ \hat{\succ}^{'} \ 3.
			\]
			Notice that to generate $\hat{\succsim}$ and in $\hat{\succsim}^{'}$, all the ties among singletons coalitions are broken in $\succsim$ and in $\succsim^{'}$.  The CAT property does not prescribe specific changes in the social ranking between parties $1$ and $3$ (whether in favor of $1$, $3$, or both being equally ranked) in  $\hat{\succsim}$ and $\hat{\succsim}^{'}$. Instead, it only asserts that a social ranking should consistently rank parties $1$ and $3$  in the same way  in both $\hat{\succsim}$ and $\hat{\succsim}^{'}$
			
		\end{example}
		
		The CAT property can be interpreted as a measure of fairness aimed to treat equally individuals in response to identical changes in  the position of coalitions in two distinct coalitional rankings: the social ranking between two individuals should not depend on the initial coalitional ranking, if they start from a social position of indifference and precisely the same tie-breaks are operated between certain coalitions. 
		
		
		The next property, introduced in \cite{bernardi:hal-02191137}, suggests giving greater importance to the positions of players within excellent coalitions, compared to positions in less favorable coalitions, in determining a social ranking.
		In this direction, if a decision is taken over the fact that an element must be considered socially strictly more relevant than another, the property of independence from the worst set states that breaking ties between coalitions in the last equivalence class should not affect such a decision.
		
		\begin{axiom}[Independence from the Worst Set (IWS)] \label{indep below}
			A solution $R$ satisfies \emph{Independence from the Worst Set}  if for any  ranking $\succsim \in \PP(\PN)$ with the associated quotient order  $\succ$ such that $$\Sigma_1 \succ  \Sigma_2  \succ  \Sigma_3  \succ  \dots \succ  \Sigma_{l}$$
			with $l\ge 2$, and $i,j\in N$ such that 	  $$i P^\succsim j,$$ then it holds: $$i P^{\succsim'} j $$
			for any partition $T_1,\dots, T_m$ of $\Sigma_l$ and for any  ranking $\succsim' \in \PP(\PN)$ with the associated quotient order $\succ'$ such that
			$$ \Sigma_1 \succ' \Sigma_2 \succ'  \dots \succ'   \Sigma_{l-1}\succ' T_1 \succ' \dots \succ' T_m.$$
		\end{axiom}
		
		\begin{example}\label{exIWT}
			Consider again the coalitional ranking  $\succsim$ from Example \ref{ex3.1bis}
			and the coalitional ranking $\hat{\succsim}$ from Example \ref{exCAT}.
			Notice that $\hat{\succsim}$ is obtained from $\succsim$  by breaking ties in the worst equivalence class of $\succsim$, while the relations involving coalitions not in the worst position in $\succsim$ remain unchanged.
			So, if for some reasons a social ranking satisfying the IWS property ranks,  in the coalitional ranking $\succsim$, a party $i \in \{1,2,3\}$ strictly better than $j\in \{1,2,3\} \setminus  \{i\}$, then it must continue to rank $i$ strictly better than $j$ also in the coalitional ranking $\hat{\succsim}$.
		\end{example}
		
		Axiom \ref{indep below} (IWS) can be interpreted as a principle aimed to reward the excellence 
		of coalitions.
		It posits that if one believes that alterations in the positions of coalitions within the worst equivalence class of a coalitional ranking (provided that the worst coalitions continue to be dominated as before the modification) play a secondary role, then the decision about a strict ranking between elements should remain unaltered.

		According to the IWS axiom the excellence is rewarded. Instead, if a decision is guided by the punishment of the mediocrity, a dual principle has to be introduced: once a total order is established for a pair of objects, it remains unchanged even if we refine the best-preferred equivalence class.
		
		\begin{axiom}[Independence from the Best Set (IBS)] \label{indep best}
			A solution $R$ satisfies \emph{Independence from the Best Set}  if for any  ranking $\succsim \in \PP(\PN)$ with the associated quotient order  $\succ$ such that $$\Sigma_1 \succ  \Sigma_2  \succ  \Sigma_3  \succ  \dots \succ  \Sigma_{l}$$
			with $l\ge 2$, and $i,j\in N$ such that 	  $$i P^\succsim j,$$ then it holds: $$i P^{\succsim'} j $$
			for any partition $T_1,\dots, T_m$ of $\Sigma_1$ and for any  ranking $\succsim' \in \PP(\PN)$ with the associated quotient order $\succ'$ such that
			$$ T_1 \succ' \dots \succ' T_m\succ' \Sigma_2 \succ'  \dots \succ'   \Sigma_{l}.$$
			
		\end{axiom}

		\begin{example}\label{exIBT}
			Consider the coalitional ranking $\succsim^{'}$ from Example \ref{exCAT}, and a new coalitional ranking $\succsim^{''}
			\in \PP(\PN)$ that emerges as a result of a slight alteration in the weights of parties only affecting coalitions in the first equivalence class and
			such that
			\[
			123 \succ^{''} 13 \succ^{''} 12 \sim^{''} 23 \succ^{''} 1 \sim^{''} 2 \sim^{''} 3 .
			\]
			
			Notice that $\succsim^{''}$ is obtained from 
			by breaking ties in the first equivalence class in $\succsim^{'}$, while the relations involving coalitions not in the first position in $\succsim^{'}$ remain unchanged.
			So, if for some reasons a social ranking satisfying the IBS property ranks a party $i \in \{1,2,3\}$ strictly better than $j\in \{1,2,3\}$, $j \neq i$, in the coalitional ranking $\succsim^{'}$, then it must rank $i$ strictly better than $j$ also in the coalitional ranking $\succsim^{''}$.
		\end{example}

		
		To introduce the next axiom, we first need to define the notion of $k$-dominance relation aimed at describing those coalitional rankings where for any coalition $S$ of a fixed cardinality $k$, $k \in \{0, \ldots, n-2\}$, and not containing neither $i$ nor $j$ we have that $S \cup \{i\}$ is at least as strong as
		$S \cup \{i\}$. 
		
		\begin{definition}[$k$-dominance relation 
			]\label{kdom}
			Consider a ranking $\succsim  \in \PP(\PN)$. For any pair of players $i,j \in N$ and any $k \in \{0, \ldots, n-2\}$, the \textit{$k$-dominance relation}
			$R^\succsim_k \subseteq N \times N$ is defined as follows:
			\[
			i R^\succsim_k j \Leftrightarrow S \cup \{i\} \succsim S \cup \{j\}  \mbox{ for all } S \subseteq N \setminus \{i,j\} \mbox{ with }|S|=k.
			\]
   Moreover, we say that $i$ $k$-dominates $j$ strictly (in notation, $i P^\succsim_k j\}$) if $i R^\succsim_k j$ and there exists $S \subseteq N \setminus \{i,j\}$ with $|S|=k$ such that $S \cup \{i\} \succ S \cup \{j\}$, while $i$ and $j$ are said indifferent in a $k$-dominance relation (denoted as $i I^\succsim_k j$) if   $S \cup \{i\} \sim S \cup \{j\}$ for all $S \subseteq N \setminus \{i,j\}$ with $|S|=k$.
   \end{definition}
			
		
Notice that if a coalitonal ranking is such that an element $i$ $k$-dominates and element $j$ for any $k \in \{0, \ldots, n-2\}$, and for some $k$ the $k$-dominance relation between $i$ and $j$ is strict ($i P^\succsim_k j$), then we have a situation where a social ranking satisfying strict desirability must rank $i$ strictly better than $j$. 

Denote by
    $\mathcal{KP}_{ij}^\succsim=\{k \in \{0, \ldots, n-2\}: i P^\succsim_k j\}$ the set of indices $k$ such that $i$ $k$-dominates $j$ strictly,
				and by
 $\mathcal{KI}_{ij}^\succsim=\{k \in \{0, \ldots, n-2\}: i I^\succsim_k j\}$ the set of indices $k$ such that $i$ and $j$ are indifferent in an $k$-dominance relation.
		The following axiom addresses dichotomous coalitional rankings where $\mathcal{KP}_{ij}^\succsim \neq \emptyset$ and $\mathcal{KP}_{ji}^\succsim \neq \emptyset$ and therefore Axiom \ref{Des2} (SDes) does not apply.

		\begin{axiom}[ $k$-Desirability on a Dichotomous ranking ($k$-DD)] \label{kdesi}
			A solution $R$ satisfies the property of $k$-\emph{Desirability on a Dichotomous ranking}  if for any dichotomous coalitional ranking $\succsim \in \PP(\PN)$, with the associated quotient order $\Sigma_1 \succ  \Sigma_2$, and for any $i,j\in N$  such that
			$$\mathcal{KP}_{ij}^\succsim \ \cup \ \mathcal{KP}_{ji}^\succsim \ \cup \ \mathcal{KI}_{ij}^\succsim=\{0,\ldots, n-2\},\qquad \mathcal{KP}_{ij}^\succsim\neq \emptyset \mbox{ and } \mathcal{KP}_{ji}^\succsim\neq \emptyset, $$
			it holds that
			$$\min\{k \in \mathcal{KP}_{ij}^\succsim\} \ < \ \min\{k \in \mathcal{KP}_{ji}^\succsim\} \ \Longrightarrow \ i P^{\succsim} j. $$
			
		\end{axiom}
		
		In dichotomous coalitional rankings where the information conveyed by the $k$-dominance relation between two element $i$ and $j$ is conflicting, the $k$-DD property advocates in favour of the strict $k$-dominance relation with the smallest value of $k$. 
		Axiom \ref{kdesi} ($k$-DD) is applied exclusively to dichotomous coalitional rankings (commonly also known as simple games), where only two equivalence classes $\Sigma_1 \succ \Sigma_2$ exist (in the framework of simple games, elements of $\Sigma_1$ are also called winning coalitions and those in $\Sigma_2$, losing coalitions). In this situations, the $k$-DD property, akin to other power indices for simple games \cite{aleandri2022minimal,aleandri2023lexicographic,deegan1978new,holler2013power}, demands that the comparison between two elements $i$ and $j$ focuses on the analysis of coalitions within the class of winning coalitions. Moreover, an element asserting the relation of $k$-dominance over winning coalitions with the smallest cardinality  is granted a better social ranking. 
  
		\begin{example}\label{exkDes}
			Consider the dichotomous coalitional  ranking $\succsim \in \PP(\PN)$ with $N=\{1,2,3\}$  such that
			\[
			123 \sim 12 \sim 23 \sim 1  \succ  13 \sim 2 \sim 3.
			\]
			Notice that element $1$ (strictly) $0$-dominates element $2$, while $2$ (strictly) $1$-dominates element $1$. So, $\mathcal{KP}_{12}=\{0\}$, $\mathcal{KP}_{21}=\{1\}$ and $\mathcal{KI}_{12}=\emptyset$. For $0=\min\{k \in \mathcal{KP}_{12}^\succsim\} \ < \ \min\{k \in \mathcal{KP}_{21}^\succsim\}=1.$, a social ranking solution satisfying the Axiom \ref{kdesi} ($k$-DD) must rank party $1$ strictly better than party $2$. Differently stated, in this example the $k$-DD property states that singletons coalitions count more than coalitons of two elements.
		\end{example}
		Axiom \ref{kdesi} ($k$-DD) is well adapted to compare the relevance in a political framework, where the formation of large coalitions may prove more challenging than creating smaller ones, due to many factors, including the existence of intricate institutional rules, the necessity for mediators in decision-making processes, increased negotiation costs, and various ``psychological'' aspects related to divergent political positions among coalition members.
		
		The next axiom derives insights by comparing two coalitional rankings obtained from an original one where two elements $i$ and $j$ initially share the same social position. 
		The first coalitional ranking is obtained by improving some coalitions  from the last equivalence class to the penultimate one. The second, undergoes the same kind of improvements of coalitions, but also all coalitons from the first to the penultimate equivalence class are placed in one single equivalence class, so producing a dichotomous coalitional ranking. The property of  consistency after indifference than specifies that the ranking between $i$ and $j$ should be coherently maintained  across the two new calitional rankings.

		\begin{axiom}[Consistency after Indifference (CI)] \label{CI}
			A solution $R$ satisfies \emph{Consistency after Indifference}  if for any  ranking $\succsim \in \PP(\PN)$ with the associated quotient order  $\succ$ such that $$\Sigma_1 \succ  \Sigma_2  \succ  \Sigma_3  \succ  \dots \succ  \Sigma_{l}$$
			with $l\ge 2$, and with $i I^\succsim j$, for $i,j\in N$, it holds that: $$i R^{\succsim'} j \Leftrightarrow i R^{\succsim''} j $$
			for any $\Sigma \subset \Sigma_l$ and for all power relations $\succsim', \succsim'' \in \PP(\PN)$ with the associated quotient order $\succ'$ and $\succ''$ such that
			$$ \Sigma_1 \succ' \Sigma_2 \succ'  \dots \succ'   \Sigma_{l-1}\cup \Sigma \succ' \Sigma_l \setminus \Sigma$$
			and
			$$ \Sigma_1 \cup \Sigma_2 \cup  \dots \cup   \Sigma_{l-1}\cup \Sigma \succ'' \Sigma_l \setminus \Sigma. $$
		\end{axiom}
		Notice that as particular case where $\succsim''$ in Axiom \ref{CI} is obtained with $\Sigma = \emptyset$, the CI property asserts that if a dichotomous ranking is derived by making the union of the first $l-1$ equivalence classes of a coalitional ranking $\succsim$ where elements $i$ and $j$ are socially indifferent, then the indifference between $i$ and $j$ must persist in the resulting dichotomous ranking $\succsim''$, since $\succsim=\succsim'$. 
		
		This property proves useful for guiding the behavior of a social ranking  under moderate changes in coalition strength within the worst equivalence class, and where a dichotomous coalitional ranking (i.e., ranking $\succsim^{''}$)  serves as a reference point for a simplified analysis of such changes, as illustrated by the following example.
  
		\begin{example}\label{exCI}
			Consider the coalitional ranking $\succsim'$ from Example \ref{exCA}, and suppose that a social ranking $R$ ranks $1$ and $2$ socially indifferent in $\succsim'$, so $1 I^{\succsim'} 2$.
			Consider the coalitional rankings $\succsim^{''}$ and $\succsim^{'''}$
			such that
			\[
			123 \sim^{''} 1 \sim^{''} 12 \succ^{''} 23 \sim^{''} 13 \succ^{''} 2 \sim^{''} 3,
			\]
			and
			\[
			123 \sim^{'''} 12 \sim^{'''} 23 \sim^{'''} 1 \sim^{'''} 13   \succ^{'''} 2 \sim^{'''} 3.
			\]
			where $\succsim^{''}$ is obtained from $\succsim^{'}$ improving the position of coalition $13$ from the last equivalence class to the penultimate one, and 
			$\succsim^{'''}$ is a dichotomous coalitional ranking obtained from $\succsim^{''}$  by making the union of the two first equivalence classes.
			The CI property does not specify the exact relation between $1$ and $2$ a social ranking should yield in the
			new coalitional rankings $\succsim^{''}$ and $\succsim^{'''}$, but it requires that in both coalitional rankings the social ranking between $1$ and $2$ must be consistent, i.e. $1 R^{\succsim^{''} } 2 \Leftrightarrow 1 R^{\succsim^{'''} } 2$.
		\end{example}

		
\section{Main results}\label{sec:main}

		\subsection{An axiomatic characterisation of the CP-majority}\label{sec:cp}
		
In this section, we explore the implications of adopting a social ranking that satisfies 	Axioms \ref{Des2} (SDes),  \ref{N} (Neu), \ref{EC} (EC)  and \ref{CAT} (CAT).
			
We start our axiomatic analysis with the following lemma showing that social rankings solutions satisfying Axioms  \ref{N} (Neu) and \ref{EC} (EC) 	consider elements supported by the same number of strict CP-comparisons as equally ranked\footnote{This fact is also shown in the proof of Theorem 1 in the paper \cite{haret:hal-02103421}.}. 	
		\begin{lemma}\label{lemcp}
			Let $R$ be a social ranking satisfying Axioms \ref{N} (Neu) and \ref{EC} (EC) . Then for any $i,j \in N$ and $\succsim \in \PP(\PN)$ such that $|D_{ij}^{\succsim}| = |D_{ji}^{\succsim}|$ we have that $i \ I^\succsim \ j$.
		\end{lemma}
		\begin{proof}
			Let $i,j \in N$ and $\succsim \in \PP(\PN)$ be such that $|D_{ij}^{\succsim}| = |D_{ji}^{\succsim}|$.
			
			Consider a bijection $\pi$ on $2^{N \setminus \{i,j\}}$ such that  $\pi(S)\in D_{ji}^{\succsim}$ for each $S \in D_{ij}^{\succsim}$ and $\pi(S)=S$ for all the remaining coalitions $S \in 2^{N \setminus \{i,j\}} \setminus (D_{ij}^{\succsim} \cup D_{ji}^{\succsim} )$.
			
			Define another ranking $\succsim'\in \PP(\PN)$
			$$S \cup \{i\} \succsim S \cup \{j\} \Leftrightarrow \pi(S) \cup \{i\} \succsim' \pi(S) \cup \{j\},$$
			for all $S \in 2^{N \setminus \{i,j\}}$. By EC, we have that
			\begin{equation}\label{ec1}
				i R^{\succsim} j \Leftrightarrow i R^{\succsim'} j.
			\end{equation}
			Notice that by construction of  $\pi$, we also have that for all $S \in 2^{N \setminus \{i,j\}}$
			$$S \cup \{i\} \succsim' S \cup \{j\} \Leftrightarrow S \cup \{j\} \succsim S \cup \{i\}.$$
			Then, by Neu, 
			\begin{equation}\label{n1}
				i R^{\succsim'} j \Leftrightarrow j R^{\succsim} i.
			\end{equation}
			
			By relations (\ref{ec1}) and (\ref{n1}), it follows that $$i R^{\succsim} j \Leftrightarrow j R^{\succsim} i,$$
			and by the fact that  $R^{\succsim}$ is a total relation, we finally obtain $i I^{\succsim} j$.
		\end{proof}

$ 1 \ P^{(\hat{\succsim}^{''} \setminus B)} \ 2$.




\noindent
We now present an axiomatic characterization of the CP-majority.  
		\begin{theorem}\label{theorem:CPaxiom1}
			CP-majority is the only social ranking  that satisfies Axioms \ref{Des2} (SDes), \ref{N} (Neu), \ref{EC} (EC) and \ref{CAT} (CAT).
		\end{theorem}
		
		\begin{proof}
			It is easy to check that $R_{CP}$ satisfies Axiom \ref{Des2} (SDes). It also satisfies Axioms \ref{N} (Neu) and \ref{EC} (EC) 
			as already shown in the literature \cite{haret:hal-02103421}. 
			
			To see that $R_{CP}$ satisfies also Axiom \ref{CAT} (CAT), recall that $\succsim$ and $\succsim'$ are subsets of the cartesian product $\mathcal{F}(N)\times \mathcal{F}(N)$, then $\succsim \setminus B$ is obtained removing the elements of $B$ from $\succsim$. Notice that, given two rankings $\succsim, \succsim'$, removing a set $B\subseteq\succsim\cap\succsim'$ such that $\succsim\setminus B$ and $\succsim'\setminus B$ remain total preorders, impact the cardinalities of $D_{ij}^{(\cdot)}$ and $D_{ji}^{(\cdot)}$ in the new ranking $\succsim\setminus B$ and $\succsim'\setminus B$ at the same extent. More precisely, let $W_{ij}=\{S \in 2^{N\setminus \{i,j\}}: (S \cup \{j\}, S \cup \{i\}) \in B\}$  the set coalitions such that the corresponding CP-comparisons in $B$ are in favour of $j$. Removing pairs in $W_{ij}$ from a coalitional ranking $\succsim$ or $\succsim'$ plays in favour of $i$ without affecting $j$, and viceversa, removing pairs in $W_{ji}$ from a coalitional ranking $\succsim$ or $\succsim'$ plays in favour of $j$ without affecting $i$, (recall that $\succsim\setminus B$ and $\succsim'\setminus B$ are total preorders). Therefore,  we have that  $|D_{ij}^{(\succsim\setminus B)}|=|D_{ij}^{\succsim}|+|W_{ij}|$, $|D_{ji}^{(\succsim\setminus B)}|=|D_{ji}^{\succsim}|+|W_{ji}|$,  $|D_{ij}^{(\succsim'\setminus B)}|=|D_{ij}^{\succsim'}+|W_{ij}|$, and $|D_{ji}^{(\succsim'\setminus B)}|=|D_{ji}^{\succsim'}|+|W_{ji}|$. 
			Now suppose that $i I^{\succsim}_{CP} j$ and $i I^{\succsim'}_{CP} j$ as demanded by Axiom \ref{cat} (CAT). Then, $|D_{ij}^{\succsim}|=|D_{ji}^{\succsim}|$ and $|D_{ij}^{\succsim'}|=|D_{ji}^{\succsim'}|$. 
			Consequently,
			\[
			i R^{(\succsim\setminus B)}_{CP} j \Leftrightarrow |D_{ij}^{\succsim}|+|W_{ij}| \geq |D_{ji}^{\succsim}|+|W_{ji}| \Leftrightarrow |W_{ij}| \geq |W_{ji}|,
			\]
			and
			\[
			i R^{(\succsim'\setminus B)}_{CP} j \Leftrightarrow |D_{ij}^{\succsim'}|+|W_{ij}| \geq |D_{ji}^{\succsim'}|+|W_{ji}| \Leftrightarrow |W_{ij}| \geq |W_{ji}|;
			\]
			so, it follows that $i R^{(\succsim\setminus B)}_{CP} j \Leftrightarrow i R^{(\succsim'\setminus B)}_{CP} j$.

			To show that $R_{CP}$ is the unique solution fulfilling Axioms \ref{Des2} (SDes), \ref{N} (Neu), \ref{EC} (EC) and \ref{CAT} (CAT) we must prove that, if a solution $R$ satisfies these axioms  then $ i R^\succsim j \Leftrightarrow i R^\succsim_{CP} j$ for all $i,j \in N$ and $\succsim \in \PP(\PN)$.
			
			Let $R$ be a solution fulfilling Axioms  \ref{Des2} (SDes), \ref{N} (Neu), \ref{EC} (EC) and \ref{CAT} (CAT) and $\succsim \in \PP(\PN)$ with the associated quotient order  $\succ$ such that $\Sigma_1 \succ   \ldots \succ  \Sigma_{l}$.
			
			We first prove that for any $i,j \in N$, $ i P^\succsim_{CP} j \Leftrightarrow i P^\succsim j$.\\
			($\Rightarrow$)\\
			Let $i P^\succsim_{CP} j$. So, $|D_{ij}^{\succsim}|>|D_{ji}^{\succsim}|$.	Define a set of coalitions $E \subseteq D_{ij}^{\succsim}$ such that $|E|=|D_{ij}^{\succsim}|-|D_{ji}^{\succsim}|$.
			
			Consider a new ranking $\succsim' \in \PP(\PN)$
			such that $D_{ij}^{\succsim'}=D_{ij}^{\succsim}\setminus E$ and  $D_{ji}^{\succsim'}=D_{ji}^{\succsim}$, and there exists an equivalence class $\Sigma$ of $\succsim'$ such that $\Sigma=\{S\cup \{i\}:S \in E\}\cup \{S\cup \{j\}:S \in E\}$. 
			By Lemma \ref{lemcp}, it follows that $i I^{\succsim'} j$.

			Now consider another ranking $\succsim'' \in \PP(\PN)$ where $|D_{ji}^{\succsim''}|=|D_{ji}^{\succsim''}|=0$ and there exists an equivalence class $\Sigma$ of $\succsim''$ such that $\Sigma=\{S\cup \{i\}:S \in E\}\cup \{S\cup \{j\}:S \in E\}$.
			By Axiom \ref{N} (Neu)  we have that $i I^{\succsim''} j$. Let $$B=\{(S \cup \{j\}, T \cup \{i\}): S, T \in E\}.$$
			Clearly, $B \subseteq \succsim' \cap \succsim''$. Moreover, removing these pairs from $\succsim'$ and $\succsim''$ we create two new relations $\succsim' \setminus B$ and $\succsim'' \setminus B$ that are still rankings in $\PP(\PN)$. In fact, only the ties between coalitions in the equivalence class $\Sigma$ of both rankings are broken in favor of $i$, and so $\Sigma$ is split in two equivalence classes $\Sigma_i=\{S \cup \{i\}: S \in E\}$ and $\Sigma_j=\{S \cup \{j\}: S \in E\}$ such that $\Sigma=\Sigma_i \cup \Sigma_j$ and in the respective quotient orders we hve $\Sigma_i \succ'\setminus B \Sigma_j$ and $\Sigma_i \succ''\setminus B \Sigma_j$.
			
			So, $(S \cup \{i\},S \cup \{j\}) \in \succ''\setminus B$ but $(S \cup \{j\},S \cup \{i\}) \notin \succ''\setminus B$, for all $S \in E$ (i.e., any $S \cup \{i\}$ with $ S \in E$ is strictly stronger than $S \cup \{j\}$ in   the ranking $\succsim''\setminus B$) while all the other coalitions $T \in 2^{N \setminus \{i,j\}} \setminus E$ are such that $T \cup \{i\} \sim T \cup \{j\}$. Therefore, by Axiom \ref{Des2} (SDes), $i P^{\succsim''\setminus B} j$ ($i$ is ranked strictly better than $j$ by the CP-majority in $\succsim''\setminus B$). So, by Axiom  \ref{CAT} (CAT) on $\succsim'$ and $\succsim''$, it follows that
			\begin{equation}\label{cat}
				i \ P^{(\succsim'\setminus B)} \ j.
			\end{equation}
			Finally, notice that 
			\[
			S \cup \{i\} \succsim S \cup \{j\} \Leftrightarrow S \cup \{i\} \ \succsim' \setminus B \ S \cup \{j\},
			\]
			for all $S \in 2^{N \setminus \{i,j\}}$; so, 
			by Axiom \ref{EC} (EC) on $\succsim$ and $\succsim'\setminus B$ (using the identity bijection on $2^{N \setminus \{i,j\}}$), we have that
			$$i R^{\succsim} j \Leftrightarrow i R^{\succsim'\setminus B} j,$$
			which finally yields  $i P^{\succsim} j$ in view of relation (\ref{cat}).
			
			($\Leftarrow$)\\
			Now, let $i P^\succsim j$. Suppose that $i I^\succsim_{CP} j$. By definition of CP-majority, $|D_{ji}^{\succsim''}|=|D_{ji}^{\succsim''}|$. So, by Lemma \ref{lemcp}, since $I$ satisfies both Axioms \ref{EC} (EC) and \ref{N} (Neu), we have $i I^\succsim j$, which yields a contradiction with $i P^\succsim j$. Since it can't even be $j P^\succsim_{CP} i$ (by the implication previously proved),  it
			must be $i P^\succsim_{CP} j$.\\ \\
			We now prove that $ i I^\succsim_{CP} j \Leftrightarrow i I^\succsim j$.\\
			($\Rightarrow$)\\
			Let $ i I^\succsim_{CP} j$. Then, again by the definition of CP-majority, $|D_{ji}^{\succsim''}|=|D_{ji}^{\succsim''}|$ and so, by Lemma \ref{lemcp}, $i I^\succsim j$.\\
			($\Leftarrow$)\\
			Let $i I^\succsim j$. As we have shown previously, $ i P^\succsim_{CP} j \Leftrightarrow i P^\succsim j$. So, it is not possible that $ i P^\succsim_{CP} j $ or $ j P^\succsim_{CP} i $ and by the fact that $R^\succsim_{CP}$ is total, it follows that  $ i I^\succsim_{CP} j$.
		\end{proof}

		\subsection{An axiomatic characterization of the lex-cel}\label{sec:lexcel}

In this section, we explore the implications of adopting a social ranking that satisfies	Axioms \ref{Des2} (SDes), \ref{Sym} (Sym), \ref{CA} (CA) and  \ref{indep below} (IWS).
			
We start our axiomatic analysis with the following lemma showing that social rankings solutions satisfying the Axioms \ref{Sym} (Sym) and  \ref{CA} (CA)	rank equal the elements that occur with the same frequency in each equivalence class of a coalitional ranking. 

		\begin{lemma}\label{lem1}
			Let $R$ be a social ranking satisfying Axioms \ref{Sym} (SYM) and  \ref{CA} (CA). Then for any $i,j \in N$ and $\succsim \in \PP(\PN)$ such that $\theta^\succsim(i) = \theta^\succsim(j)$ it holds that $i \ I^\succsim \ j$.
		\end{lemma}
		\begin{proof}
			Since $\theta^\succsim(i) = \theta^\succsim(j)$, then $i_k=j_k$ for all $k \in \{1, \ldots, n\}$. Define a bijection $\pi$  on $2^{N \setminus\{i,j\}}$ such that $\pi(S) \cup \{j\} \in \Sigma_k$ for each $k \in \{1, \ldots, n\}$ and for each coalition $S \in 2^{N \setminus\{i,j\}}$   with $S \cup \{i\} \in \Sigma_k$.
			Consider a power relation $\succsim' \in \PP(\PN)$ such that $S \cup \{i\} \succsim T \cup \{j\} \Leftrightarrow \pi(S) \cup \{i\} \succsim' T \cup \{j\}$.
			Notice that $T \cup \{i\} \sim' T \cup \{j\}$ for all $T \in 2^{N \setminus\{i,j\}}$. So, by Axiom \ref{Sym} (Sym) , $i I^{\succsim'} j$ and, by Axiom $\ref{CA}$ (CA) on $\succsim$ and $\succsim'$, it holds that $i I^\succsim j$.
		\end{proof}

\noindent
We now present an axiomatic characterization of the lex-cel.
		
		\begin{theorem}\label{1}
			The lex-cel solution $R_{\ran}$ is the unique social ranking  fulfilling Axioms  \ref{Des2} (SDes), \ref{Sym} (SYM), \ref{CA} (CA) and  \ref{indep below} (IWS).
		\end{theorem}

\begin{proof}
		We first show that   $R_{\ran}$ satisfies  Axioms  \ref{Sym} (SYM).
			Let $i,j \in N$ and $\succsim \in \PP(\PN)$. Suppose that  $S \cup \{i\} \sim S \cup \{j\}$ for all  $S \subseteq N \setminus \{i,j\}$. So, $\theta^\succsim(i) = \theta^\succsim(j)$ and $i I_{\ran}^\succsim j$. 
   
To prove that $R_{\ran}$ satisfies  Axioms  \ref{Des2}
(SDes), suppose that $S \cup \{i\} \sim S \cup \{j\}$ for all  $S \subseteq N \setminus \{i,j\}$ and $T \cup \{i\} \succ T \cup \{j\}$ for some  $T \in 2^{N \setminus \{i,j\}}$. Consider the associated quotient order  $\succ$ such that $\Sigma_1 \succ  \Sigma_2   \succ  \dots \succ  \Sigma_{l}$ and let $s \in \{1, \ldots, l\}$, the smallest index such that there exists some $T \cup \{i\}\in \Sigma_s$ with $T \cup \{i\} \succ T \cup \{j\}$. Then  $i_k=j_k$ for $k \in \{1, \ldots, s-1\}$ and $i_s>j_s$ which implies $i P_{\ran}^\succsim j$.

			As already proved in Theorem 1 in \cite{bernardi:hal-02191137},  $R_{\ran}$ also satisfies Axioms  \ref{CA} (CA) and  \ref{indep below} (IWS).\\
			
			We now prove that the lex-cel is the unique social ranking satisfying Axioms \ref{Des2} (SDes), \ref{Sym} (SYM), \ref{CA} (CA) and  \ref{indep below} (IWS).
			
			Take a solution $R$ satisfying Axioms \ref{Des2} (SDes), \ref{Sym} (SYM), \ref{CA} (CA) and  \ref{indep below} (IWS) then we show that $ i R^\succsim j \Leftrightarrow i R^\succsim_{\ran} j$ for all $i,j \in N$ and $\succsim \in \PP(\PN)$.
			
			Fix  $\succsim \in \PP(\PN)$ with the associated quotient order  $\succ$ such that $\Sigma_1 \succ   \ldots \succ  \Sigma_{l}$ and $i,j\in N$.
			The first step of the proof consists in showing that $ i P^\succsim_{\ran} j \Leftrightarrow i P^\succsim j$.\\
			($\Rightarrow$)\\
			Suppose $i P^\succsim_{\ran} j$. Let $s$ be the smallest integer in $\{1, \ldots, l\}$  with $i_s > j_s$. Define a set of coalitions $B=\{S \in 2^{N \setminus\{i,j\}}: S \cup \{i\} \in \Sigma_s\}$ with $|B|=i_s-j_s>0$.
			
			Define a new coalitional ranking $\succsim'$ with the associated quotient order  $\succ'$ such that $\Sigma_1 \succ'  \ldots \succ' \Sigma_s \succ' T_{s+1}$ such that $T_{s+1}= \Sigma_{s+1} \cup \ldots \cup \Sigma_{l}$. By Axiom \ref{indep below} (IWS),
			\begin{equation}\label{conindep}
				i P^{\succsim'} j \Rightarrow i P^{\succsim} j.
			\end{equation}
			Define a bijection $\pi$  on $2^{N \setminus\{i,j\}}$ such that $\pi(S) \cup \{j\} \in \Sigma_k$ for each $k \in \{1, \ldots, s\}$ and each coalition $S \in 2^{N \setminus\{i,j\}} \setminus B$ with $S \cup \{i\} \in \Sigma_k$, and such that $\pi(S) \cup \{j\} \in T_{s+1}$, for each $S \in B$.
			Let $\succsim''$ be a power relation such that $S \cup \{i\} \succsim' T \cup \{j\} \Leftrightarrow \pi(S) \cup \{i\} \succsim'' T \cup \{j\}$. Notice that $i P^{\succsim''}_d j$. So, by the fact that $R$ fulfils Axiom \ref{Des2} (SDes),it holds that $i P^{\succsim''} j$. Moreover, by Axiom \ref{CA} (CA), $i R^{\succsim'} j \Leftrightarrow i R^{\succsim''} j$ and by relation (\ref{conindep}) it follows that $i P^{\succsim} j$.\\
			($\Leftarrow$)\\
			Now, suppose $i P^\succsim j$. If $i I^\succsim_{\ran} j$, by definition of lex-cel, $\theta^\succsim(i)=\theta^\succsim(j)$. So, by Lemma \ref{lem1}, $i I^\succsim j$, which yields a contradiction with $i P^\succsim j$. Since it can't even be $j P^\succsim_{\ran} i$ (by the implication previously proved),  it
			must be $i P^\succsim_{\ran} j$.\\ \\
			The second step of the proof consists in showing that  $ i I^\succsim_{\ran} j \Leftrightarrow i I^\succsim j$.\\
			($\Rightarrow$)\\
			Suppose $ i I^\succsim_{\ran} j$. Then, again by the definition of the lex-cel, $\theta^\succsim(i)=\theta^\succsim(j)$ and so, by Lemma \ref{conindep}, $i I^\succsim j$.\\
			($\Leftarrow$)\\
			Now, suppose $i I^\succsim j$. As previously shown, $ i P^\succsim_{\ran} j \Leftrightarrow i P^\succsim j$. So, it is not possible that $ i P^\succsim_{\ran} j $ or $ j P^\succsim_{\ran} i $ and by the fact that $R^\succsim_{\ran}$ is total, it follows that  $ i I^\succsim_{\ran} j$.
			
		\end{proof}

			In \cite{bernardi:hal-02191137} the \emph{dual-lex} solution has been introduce as the dual solution of the lex-cel solution.  Given two $l$-dimensional vectors $\mathbf{i}, \mathbf{j}$ we define the lexicographic* order $\geq_{L^*}$ as
			$$
			\mathbf{i} \geq_{L^*}\mathbf{j} \quad \text { if either } \mathbf{i}=\mathbf{j} \quad \text { or } \quad \exists s: i_k=j_k, k=s+1, \ldots, l \quad \text { and } i_s<j_s .
			$$
   
			\begin{definition}[dual-lex \cite{bernardi:hal-02191137}]
				The dual-lexicographic (dual-lex) solution is the function $R_{\dual}: \PP(\PN)\to \bm{\mathcal{B}}(N)$ defined for any  ranking $\succsim \in \PP(\PN)$ as
				$$
				i R_{\dual}(\succsim) j\qquad \text { if } \qquad \theta_{\succsim}(i) \geq_{L^*} \theta_{\succsim}(j) .
				$$
			\end{definition}
			
			The following characterization for the \emph{dual-lex} holds. 
			
			\begin{theorem}\label{1dual}
				The dual-lex solution $R_{\dual}$ is the unique solution fulfilling Axioms  \ref{Des2} (SDes), \ref{Sym} (SYM), \ref{CA} (CA) and  \ref{indep best} (IBS).
			\end{theorem}
			\begin{proof}
				It is easy to observe that the dual-lex solution $R_{\dual}$ satisfies the Axiom \ref{Des2} (SDes).  Fix  $\succsim \in \PP(\PN)$ with the associated quotient order  $\succ$ such that $\Sigma_1 \succ   \ldots \succ  \Sigma_{l}$ and $i,j\in N$. If $S \cup \{i\} \succsim S \cup \{j\}$ for all $S \in 2^N\setminus \{i,j\}$ and there exists a coalition $T \in 2^N\setminus \{i,j\}$ such that $T \cup \{i\} \succ T \cup \{j\}$. Let $s\in\{1,\ldots,l\}$ the highest index  such that there exists a coalition $\tilde{T} \in 2^N\setminus \{i,j\}$ such that $\tilde{T} \cup \{i\} \succ \tilde{T} \cup \{j\}$  and $\tilde{T}\cup\{j\}\in\Sigma_s$, then  $i_s<j_s$ and $i_k=j_k$, $\forall k>s$. It follows that $\theta_{\succsim}(i) \geq_{L^*} \theta_{\succsim}(j)$.\\
				The rest of the proof of the theorem follows the same arguments of the proof of Theorem \ref{1} and they are omitted. 
			\end{proof}

		\subsection{An axiomatic characterization of the $L^{(1)}$ solution}\label{sec:lex1}

In this section, we explore the implications of adopting a social ranking that satisfies	Axioms  \ref{Des2} (SDes), \ref{Sym} (SYM), \ref{indep below} (IWS), \ref{kdesi} ($k$-DD),  \ref{PCA} (PCA) and \ref{CI} (CI).
			
We start our axiomatic analysis with the following lemma showing that social rankings solutions satisfying the Axioms \ref{Sym} (SYM) and  \ref{PCA} (PCA)	ranks equally relevant those elements that occur with the same frequency and within coalitions of the same size in each equivalence class of a coalitional ranking  (this result was also shown in paper \cite{algaba:hal-03388789}). 	

\begin{lemma}\label{lem2}
		Let $R$ be a social ranking satisfying Axioms \ref{Sym} (SYM) and  \ref{PCA} (PCA). Then for any $i,j \in N$ and $\succsim \in \PP(\PN)$ such that $M^{\succsim,i}_{sk}=M^{\succsim,j}_{sk}$ for all $s\in\{1,\ldots,n\}$ and $k \in \{1, \ldots, l\}$ it holds that $i \ I^\succsim \ j$.
	\end{lemma}
	\begin{proof}
Define a bijection $\pi$  on $2^{N \setminus\{i,j\}}$ such that $\pi(S) \cup \{j\} \in \Sigma_k$ if $S \cup \{i\} \in \Sigma_k$, for each $k \in \{1, \ldots, l\}$, for each coalition $S \in 2^{N \setminus\{i,j\}}$   with  $|\pi(S)|=|S|$.\\
		Consider a coalitional ranking $\succsim' \in \PP(\PN)$ such that
		\begin{equation*}
			S \cup \{i\} \succsim T \cup \{j\} \Leftrightarrow \pi(S) \cup \{i\} \succsim' T \cup \{j\}.
		\end{equation*}
		Notice that $T \cup \{i\} \sim' T \cup \{j\}$ for all $T \in 2^{N \setminus\{i,j\}}$. So, by Axiom \ref{Sym} (SYM), it holds that $i I^{\succsim'} j$ and, by Axiom \ref{PCA} (PCA) on $\succsim$ and $\succsim'$, it follows that $i I^\succsim j$.
	\end{proof}

\noindent
We now introduce an axiomatic characterization of the $L^{(1)}$ solution.
	
	\begin{theorem}\label{teo:L1}
		The solution $R_{L^{(1)}}^\succsim$ is the unique social ranking fulfilling Axioms \ref{Des2} (SDes), \ref{Sym} (SYM), \ref{indep below} (IWS), \ref{kdesi} ( $k$-DD), \ref{PCA} (PCA), and \ref{CI} (CI).
	\end{theorem}
	\begin{proof}
		$L^{(1)}$ solution satisfies Axiom \ref{indep below} (IWS) and \ref{PCA} (PCA) as showed in \cite{algaba:hal-03388789}. \\
		Suppose that $S\cup\{i\}\sim S\cup\{j\}$, for all $S\subseteq N\setminus\{i,j\}$ then $S\cup\{i\},S\cup\{j\}$ belong to the same equivalence class and then $M^{\succsim, i}_{sk} = M^{\succsim, j}_{sk}$ for all $s \in \{1, \ldots, n\}$, and all $k \in \{1, \ldots, l\}$. This observation implies that $R_{L^{(1)}}$ satisfies Axiom \ref{Sym} (SYM). Moreover, let $\bar{k}\in\{1,\ldots,l\}$ the smallest index such that there is $\bar{S}\in\Sigma_{\bar{k}}$ such that $\bar{S}\cup\{i\}\succ \bar{S}\cup\{j\}$, where $\bar{S}$ has minimum cardinality between the coalitions in $\Sigma_{\bar{k}}$. Then $M^{\succsim, i}_{\bar{s}\bar{k}} > M^{\succsim, j}_{\bar{s}\bar{k}}$, where $\bar{s}=|\bar{S}|$, and $i P_{L^{(1)}}^\succsim j$. This observation implies that $R_{L^{(1)}}$ satisfies Axiom \ref{Des2} (SDes).\\
		Call $s^*_i=\min\{s \in \mathcal{KP}_{ij}^\succsim\}$. For a dichotomous power relation it holds
		\begin{equation*}
			M^{\succsim, i}_{s 1} = M^{\succsim, j}_{s 1},\ \forall s<s^*_i \mbox{ and } M^{\succsim, i}_{s_i^*1} > M^{\succsim, j}_{s_i^*1},
		\end{equation*}
		then $R_{L^{(1)}}^\succsim$ satisfies Axiom \ref{kdesi} ($k$-DD). \\
		Let $\succsim\in\PP(\PN)$ such that $iI_{CP}^\succsim j$, then by definition $M^{\succsim, i}=M^{\succsim, j}$. Assume for contradiction, that there are two power relations $\succsim', \succsim'' \in \PP(\PN)$ with the associated quotient order $\succ'$ and $\succ''$ with
		$$ \Sigma_1 \succ' \Sigma_2 \succ'  \dots \succ'   \Sigma_{l-1}\cup \Sigma \succ' \Sigma_l \setminus \Sigma$$
		and
		$$ \Sigma_1 \cup \Sigma_2 \cup  \dots \cup   \Sigma_{l-1}\cup \Sigma \succ'' \Sigma_l \setminus \Sigma. $$
		such that $iP^{\succsim'}_{L^{(1)}}j$ and $jP^{\succsim''}_{L^{(1)}}i$. By $iP^{\succsim'}_{L^{(1)}}j$, there is an index $\bar{s}$ such that
		$M^{\succsim', i}_{\bar{s} l-1}>M^{\succsim', j}_{\bar{s}l-1}$ and such that $M^{\succsim', i}_{s l-1}=M^{\succsim', j}_{s l-1}$, if $\bar{s}>1$, $\forall s < \bar{s}$ . This implies that the number of coalitions containing either $i$ or $j$ of size less $\bar{s}$ in $\Sigma_1 \cup \Sigma_2 \cup  \dots \cup   \Sigma_{l-1}\cup \Sigma$ is the same. This observation contradicts the assumption $jP^{\succsim''}_{L^{(1)}}i$ and  then the $L^{(1)}$  solutionsatisfies Axiom \ref{CI} (CI).\\

We now prove that the $L^{(1)}$ solution is the unique social ranking satisfying Axioms \ref{Des2} (SDes), \ref{Sym} (SYM), \ref{indep below} (IWS), \ref{kdesi} ($k$-DD), \ref{PCA} (PCA), and \ref{CI} (CI).

		Let $R$ be a solution satisfying Axioms  \ref{Des2} (SDes), \ref{Sym} (SYM), \ref{indep below} (IWS), \ref{kdesi} ($k$-DD), \ref{PCA} (PCA), and \ref{CI} (CI) then we show that $iR^\succsim j\iff i R_{L^{(1)}}^\succsim j $, $\forall i,j\in N$ and  $\forall \succsim\in \PP(\PN)$.\\
		Take $\succsim\in \PP(\PN)$ with associated quotient order $\succ$,
		\begin{equation*}
			\Sigma_1 \succ \Sigma_2 \succ \ldots \succ \Sigma_l.
		\end{equation*}
		We start proving that $iP_{L^{(1)}}^\succsim j\iff i P^\succsim j $\\
		$(\Rightarrow)$\\
        Suppose $iP_{L^{(1)}}^\succsim j$, then  there are $\hat{k}\in\{1,\ldots,l\}$ and $\hat{s}\in\{1,\ldots,n\}$ such that
		\begin{align*}
            & M^{\succsim, i}_{\hat{s}\hat{k}} > M^{\succsim, j}_{\hat{s}\hat{k}}, \mbox{ and}\\
			\mbox{ if } \hat{k}>1,\quad & M^{\succsim, i}_{sk} = M^{\succsim, j}_{sk} \mbox{ for all } s \in \{1, \ldots, n\} \mbox{ and all } k \in \{1, \ldots, \hat{k}-1\}, \\
			\mbox{ if } \hat{s}>1,\quad & M^{\succsim, i}_{s\hat{k}} = M^{\succsim, j}_{s\hat{k}} \mbox{ for all } s \in \{1, \ldots, \hat{s}-1\}.                
		\end{align*}
		Moreover, if $\hat{s}<n$, two cases must be distinguished:
		\begin{enumerate}
			\item[Case a)]  $M^{\succsim, i}_{s\hat{k}} \geq M^{\succsim, j}_{s\hat{k}}$ for all $s \in \{\hat{s}+1, \ldots, n\}$,
			\item[Case b)]  for all  $s \in \{\hat{s}+1, \ldots, n\}$ either $M^{\succsim, i}_{s\hat{k}} \geq M^{\succsim, j}_{s\hat{k}} $ or  $M^{\succsim, i}_{s\hat{k}} \leq M^{\succsim, j}_{s\hat{k}}$.
		\end{enumerate}
		Suppose, without loss of generality, that $l\geq3$, $2\leq\hat{k}\leq l-1$ and $2\leq\hat{s}\leq n-1$.
		Consider the ranking $\succ'$ with associated quotient order
		\begin{equation*}
			\Sigma_1 \succ' \ldots \succ' \Sigma_{\hat{k}}\succ' \underbrace{\Sigma_{\hat{k}+1}\cup \ldots \cup \Sigma_{l}}_{T_{\hat{k}+1}}.
		\end{equation*}
		
		By Axiom \ref{indep below} (IWS)
		\begin{equation} \label{eq:IWS}
			i P^{\succsim'} j \Rightarrow i P^\succsim j.
		\end{equation}
		Define $B_{\hat{s}}^{ij}=\{S:\ S\cup\{i\}\in\Sigma_{\hat{k}},\ |S\cup\{i\}|=\hat{s}\}$ such that $|B_{\hat{s}}^{ij}|=M^{\succsim, i}_{\hat{s}\hat{k}} - M^{\succsim, j}_{\hat{s}\hat{k}}$. Take a bijection $\pi$ on $2^{N\setminus\{i,j\}}$ such that $|\pi(S)|=|S|$, $\forall S\in 2^{N\setminus\{i,j\}}$, and
		\begin{align*}
			& \pi(S)\cup\{j\}\in\Sigma_k, \mbox{ if } S\cup\{i\}\in\Sigma_k, \forall k<\hat{k};                                        \\
			& \pi(S)\cup\{j\}\in\Sigma_{\hat{k}}, \mbox{ if } S\cup\{i\}\in\Sigma_{\hat{k}}\setminus B_{\hat{s}}^{ij}, \mbox{ with } |S\cup\{i\}|\leq \hat{s}; \\
			& \pi(S)\cup\{j\}\in T_{k+1}, \mbox{ if } S\in B_{\hat{s}}^{ij}.
		\end{align*}
		Moreover, for all $s \in \{\hat{s}+1, \ldots, n\}$,
		\begin{enumerate}
			\item[Case a)]  define $B_{s}^{ij}=\{S:\ S\cup\{i\}\in\Sigma_{\hat{k}},\ |S\cup\{i\}|=s\}$ such that $|B_s^{ij}|=M^{\succsim, i}_{s\hat{k}} - M^{\succsim, j}_{s\hat{k}}$ then 
   
            $\pi(S)\cup\{j\}\in\Sigma_{\hat{k}}$, if $S\cup\{i\}\in\Sigma_{\hat{k}}\setminus B_{s}^{ij}$, $|S|=s$  \\
            $\pi(S)\cup\{j\}\in T_{\hat{k}+1}$, if $S\in B_{s}^{ij}$;
			
            \item[Case b.1)] Suppose $M^{\succsim, i}_{s\hat{k}} - M^{\succsim, j}_{s\hat{k}} = b \geq 0 $, define $B_{s}^{ij}=\{S:\ S\cup\{i\}\in\Sigma_{\hat{k}},\ |S\cup\{i\}|=s\}$ such that $|B_s^{ij}|=b$, then
			
			$\pi(S)\cup\{j\}\in\Sigma_{\hat{k}}$, if $S\cup\{i\}\in\Sigma_{\hat{k}}\setminus B_{s}^{ij}$, $|S|=s$  \\
			$\pi(S)\cup\{j\}\in T_{\hat{k}+1}$, if $S\in B_{s}^{ij}$;
			
			\item[Case b.2)]  Suppose $M^{\succsim, j}_{s\hat{k}} - M^{\succsim, i}_{s\hat{k}} = b \geq 0 $, define $B_{s}^{ji}=\{S:\ S\cup\{j\}\in\Sigma_{\hat{k}},\ |S\cup\{j\}|=s\}$ such that $|B_s^{ji}|=b$, then
			
			$\pi(S)\cup\{j\}\in\Sigma_{\hat{k}}$, if $S\cup\{i\}\in\Sigma_{\hat{k}}\setminus B_{s}^{ji}$, $|S|=s$ \\
			$\pi(S)\cup\{i\}\in T_{\hat{k}+1}$, if $S\in B_{s}^{ji}$.
		\end{enumerate}
		Consider the ranking $\succsim''$ such that
		\begin{enumerate}
			\item[Case a)] 	 $S\cup\{i\} \succ' T\cup\{j\} \iff \pi(S)\cup\{i\} \succ'' T\cup\{j\}$ ,
			\item[Case b.1)]  $S\cup\{i\} \succ' T\cup\{j\} \iff \pi(S)\cup\{i\} \succ'' T\cup\{j\}$,
			\item[Case b.2)]  $S\cup\{i\} \prec' T\cup\{j\} \iff \pi(S)\cup\{i\} \prec'' T\cup\{j\}$.
		\end{enumerate}
		with associated quotient
		\begin{equation*}
			\tilde{\Sigma}_1 \succ'' \ldots \succ'' \tilde{\Sigma}_{\hat{k}}\succ'' \tilde{T}_{\hat{k}+1}.
		\end{equation*}
		Then by Axiom \ref{PCA} (PCA)
		\begin{equation} \label{eq:PCA}
			i R^{\succsim'} j \iff i R^{\succsim''} j.
		\end{equation}
		Take $\Sigma=\{\pi(S)\cup\{i\},\ S\in \cup_s B_{s}^{ij}\}\cup \{S\cup\{j\},\ S\in \cup_s B_{s}^{ji}\}$ and define the ranking $\succsim'''$ as follows
		\begin{equation*}
			\tilde{\Sigma}_1 \succ''' \ldots \succ''' \tilde{\Sigma}_{\hat{k}}\setminus \Sigma \succ''' \tilde{T}_{k+1}\cup \Sigma.
		\end{equation*}
		Notice that in $\succsim'''$, for all $S\in2^{N\setminus\{i,j\}}$, it holds $S \cup \{i\}\sim S \cup \{j\} S$ then Axiom \ref{Sym} (SYM) implies $i I^{\succsim'''} j$.\\
		Now consider the ranking $\succsim''''$ with associated quotient order
		\begin{equation*}
			\tilde{\Sigma}_1 \cup \ldots \cup \tilde{\Sigma}_{\hat{k}}\succ'''' \tilde{T}_{k+1}.
		\end{equation*}
		
		Then by Axiom \ref{CI} (CI)
		\begin{equation} \label{eq:CI}
			i R^{\succsim''} j \iff i R^{\succsim''''} j.
		\end{equation}
		Using the construction of $\pi$  notice that in $\succsim''''$
		\begin{enumerate}
			\item[Case a)] 	 Axiom  \ref{Des2} (SDes) implies $i P^{\succsim''''} j$ ,
			\item[Case b)]    Axiom  \ref{kdesi} ($k$-DD) implies $i P^{\succsim''''} j$, indeed
			\begin{itemize}
				\item $\mathcal{KI}_{ij}^{\succsim''''}=\{0,\ldots,\hat{s}-1\}$ and for any $s\in\{\hat{s},\ldots,n\}$ then  either $S\cup\{i\}\succsim''''S\cup\{j\}$, $\forall S\in2^{N\setminus\{i,j\}}$ and $|S|=s$, or $S\cup\{j\}\succsim''''S\cup\{i\}$, $\forall S\in2^{N\setminus\{i,j\}}$ and $|S|=s$; \\
				\item $\hat{s}\in\mathcal{KP}_{ij}^{\succsim''''}$ and  $\mathcal{KP}_{ji}^{\succsim''''}\cap\{\hat{s}+1,\ldots,n\}\neq \emptyset$, otherwise $i$ would be strictly more desirable than $j$;
				\item $\hat{s}=\min\{k \in \mathcal{KP}_{ij}^{\succsim''''}\} \ < \ \min\{k \in \mathcal{KP}_{ji}^{\succsim''''}\}$.
			\end{itemize}
		\end{enumerate}
		Then
		\begin{equation}
			i P^{\succsim''''} j \xRightarrow{\eqref{eq:CI}} i P^{\succsim''} j \xRightarrow{\eqref{eq:PCA}} i P^{\succsim'} j \xRightarrow{\eqref{eq:IWS}} i P^\succsim j.
		\end{equation}
		($\Leftarrow$)\\
		Suppose $i P^\succsim j$. If $i I^\succsim_{L^{(1)}} j$, by definition of $L^{(1)}$  solution , $M^{\succsim,i} =M^{\succsim,j}$. So, by Lemma \ref{lem2}, $i I^\succsim j$, which yields a contradiction with $i P^\succsim j$. Since it can't even be $j P^\succsim_{L^{(1)}} i$ (by the implication previously proved),  it must be $i P^\succsim_{L^{(1)}} j$.\\ 
		We conclude the proof showing that $ i I^\succsim_{L^{(1)}} j \Leftrightarrow i I^\succsim j$.\\
		($\Rightarrow$)\\
		Suppose $ i I^\succsim_{L^{(1)}} j$. Then, again by the definition of the $L^{(1)}$  solution, $M^{\succsim,i} =M^{\succsim,j}$ and so, by Lemma \ref{lem2}, $i I^\succsim j$.\\
		($\Leftarrow$)\\
		Suppose $i I^\succsim j$. As previously shown, $ i P^\succsim_{L^{(1)}} j \Leftrightarrow i P^\succsim j$. So, it is not possible that $ i P^\succsim_{L^{(1)}} j $ or $ j P^\succsim_{L^{(1)}} i $ and by the fact that $R^\succsim_{L^{(1)}}$ is total, it follows that  $ i I^\succsim_{L^{(1)}} j$.
	\end{proof}
	
		In \cite{algaba:hal-03388789} the authors mention the possibility to define a dual solution of the $L^{(1)}$ solution using the Independence of the Best Set axiom in place of the Independence of the worst Set axiom. In that paper they do not give a complete axiomatization because they need to introduce new formulation of axioms about the path monotonicity. Using, instead, Axiom \ref{Des2} (SDes) the characterization follows as well.
		\begin{definition}[$L^{(1)}_*$ solution \cite{algaba:hal-03388789}] \label{l1sol-dual}
    			The $L^{(1)}_*$ solution is the map $R_{L^{(1)}_*}^\succsim:\PP(\PN)\to \bm{\mathcal{B}}(N)$
			such that for any  power relation $\succsim \in \PP(\PN)$ and $i,j \in N$:
			$$
			i \ I_{L^{(1)}_*}^\succsim \ j \Longleftrightarrow M^{\succsim, i}=M^{\succsim, j}
			$$
			and
			$$
			i \ P_{L^{(1)}_*}^\succsim \ j \Longleftrightarrow
			\begin{array}{l}
				\exists \hat{s} \in \{1, \ldots, n\} \mbox {and } \hat{k} \in  \{2, \dots, l\} \mbox{ such that } \\
			\qquad \left\{\begin{array}{ll}
				M^{\succsim, i}_{sk} = M^{\succsim, j}_{sk},                         & \forall s \in \{1, \ldots, n\}, \forall k \in \{\hat{k}+1, \ldots, l\}; \smallskip \\
				M^{\succsim, i}_{s\hat{k}} = M^{\succsim, j}_{s\hat{k}},             & \forall s \in \{1,\ldots,\hat{s}-1\}; \smallskip                                 \\
				M^{\succsim, i}_{\hat{s}\hat{k}} < M^{\succsim, j}_{\hat{s}\hat{k}}. &
			\end{array}\right.
		\end{array}
		$$
	\end{definition}
	
	Next theorem gives a characterization of the $L^{(1)}_*$, the dual solution of \emph{$L^{(1)}$}.
	
	\begin{theorem}
		The solution $R_{L^{(1)}_*}^\succsim$ is the unique solution fulfilling Axioms \ref{Des2} (SDes), \ref{Sym} (SYM), \ref{indep best} (IBS), \ref{kdesi} ( $k$-DD), \ref{PCA} (PCA), and \ref{CI} (CI).
	\end{theorem}
	\begin{proof}
		It is easy to observe that the$L^{(1)}_*$  solutionsatisfies the Axiom \ref{Des2} (SDes).  Fix  $\succsim \in \PP(\PN)$ with the associated quotient order  $\succ$ such that $\Sigma_1 \succ   \ldots \succ  \Sigma_{l}$ and $i,j\in N$. Suppose that $S \cup \{i\} \succsim S \cup \{j\}$ for all $S \in 2^N\setminus \{i,j\}$ and there exists a coalition $T \in 2^N\setminus \{i,j\}$ such that $T \cup \{i\} \succ T \cup \{j\}$. Let $\hat{k}$ the highest index such that there exists at least one coalition $\hat{T} \in 2^N\setminus \{i,j\}$ with $\hat{T} \cup \{i\} \succ \hat{T} \cup \{j\}\in\Sigma_{\hat{k}}$. Moreover, suppose that $\hat{T}$ has the lowest cardinality between the coalitions that are in $\Sigma_{\hat{k}}$ and that satisfy the previous relation. Then $M^{\succsim, i}_{sk} = M^{\succsim, j}_{sk}$, $\forall s \in \{1, \ldots, n\}, \forall k \in \{1,\ldots,\hat{k}-1\}$, $M^{\succsim, i}_{s\hat{k}} = M^{\succsim, j}_{s\hat{k}}$, $\forall s \in \{\hat{s}+1, \ldots, n \}$ and    $M^{\succsim, i}_{\hat{s}\hat{k}} < M^{\succsim, j}_{\hat{s}\hat{k}}$. From this observation it follows $	i \ P_{L^{(1)}_*}^\succsim \ j$. \\
		The proof of the theorem proceeds in a manner analogous to that of Theorem \ref{teo:L1} and is therefore omitted. 
	\end{proof}

\section{Comparison between the axioms}\label{sec:logic}

In this section we study the dependence between the axioms introduced in Section \ref{sec:axioms}. Moreover we show the relation between the strictly desirability and the monotonicity used in \cite{bernardi:hal-02191137}.

\begin{proposition}\label{logic1}
	The following statements  hold true:
	\begin{itemize}
		\item[i)] Neutrality axiom \ref{N} (Neu) $\Rightarrow$ Symmetry axiom \ref{Sym} (Sym), but the opposite implication does not hold;
		\item[ii)] There is no logical dependence between Equality of Coalitions axiom \ref{EC} (EC) and Coalitional Anonymity axiom \ref{CA} (CA).
		\item[iii)] There is no logical dependence between Equality of Coalitions axiom \ref{EC} (EC) and Per-size Coalitional Anonymity axiom \ref{PCA} (PCA).
	\end{itemize}
\end{proposition}
Before proving the above proposition it is useful to recall that in \cite{algaba:hal-03388789} it is shown that Coalition Anonimity implies Per-size Coalition Anonymity.
\begin{proof}
	Statement $i)$.\\
	It is easy to observe that Neu $\Rightarrow$  Sym. To check that the opposite implication does not hold, consider a solution $R:\bm{\mathcal{R}}(\PN)\to \bm{\mathcal{B}}(N)$ such that, for any $i,j \in N$, $i P^{\succsim} j$ if $\{i\} \succ \{j\}$, and $i I^{\succsim} j$, otherwise. Clearly, $R^{\succsim}$ satisfies Sym, but not Neu.\\
	
	Statement (ii)\\
	Consider the CP-majority relation $R_{CP}$. It satisfies EC, but not CA. In fact, consider a coalitonal ranking $\succsim \in \PP(\PN)$ with $N=\{i,j,k\}$ and such that
	$j \succ ik \succ jk \sim i$. We have that $ i I^{\succsim}_{CP} j$. Now, consider a bijection on $\{\emptyset,k\}$ such that $\pi(\emptyset)=k$ and a ranking $\sqsupseteq \in \PP(\PN)$ such that $j \sqsupset i \sqsupset jk \sim ik$. A solution that satisfies CA on $\succsim$ and $\sqsupseteq$ would give the same ranking between $i$ and $j$ over the two coalitional rankings. But we have $ j P^{\sqsupseteq}_{CP} i$, so $R_{CP}$ does not saisfy CA.
	
	Now, consider the lex-cel solution $R_{\ran}$. The lex-cel satisfies CA, but not EC. In fact, consider for instance the coalitional ranking $\succsim \in \PP(\PN)$ with $N=\{i,j,k\}$ and such that
	$ik \succ jk \succ j \succ i$, and another ranking $\sqsupseteq$ such that $j \sqsupset i \sqsupset ik \sqsupset jk$. A solution that satisfies EC on $\succsim$ and $\sqsupseteq$  (with the identity bijection in Axiom \ref{EC}) should rank $i$ and $j$ in the same way in both coalitional rankings $\succsim$ and $\sqsupseteq$. Indeed coalition $ik$ is better then coalition $jk$ in both $\succsim$ and $\sqsupseteq$ and coalition $j$ is better then coalition $j$ in both $\succsim$ and $\sqsupseteq$. On the other hand the lex-cel yields $i P^{\succsim}_{\ran} j$
	and $j P^{\sqsupseteq}_{\ran} i$. So, lex-cel does not satisfy EC. \\
	
	Statement (iii)\\ 
    Consider the CP-majority relation $R_{CP}$. It satisfies EC, but not PCA. In fact, consider a coalitonal ranking $\succsim \in \PP(\PN)$ with $N=\{i,j,k,l\}$ and such that
	$jl \succ ik \succ jk \sim il\succ \mathcal{C} $, where $\mathcal{C}$ is an equivalence class containing all the other coalitions not already displayed. We have that $ i I^{\succsim}_{CP} j$. Now, consider a size-invariant bijection on $\{\emptyset,k,l,kl\}$ such that $\pi(k)=l$ and $\pi(l)=k$ and a ranking $\sqsupseteq \in \PP(\PN)$ such that $jl \sqsupset il \sqsupset jk \sim ik$. A solution that satisfies PCA on $\succsim$ and $\sqsupseteq$ would give the same ranking between $i$ and $j$ over the two coalitional rankings. But we have $ j P^{\sqsupseteq}_{CP} i$, so $R_{CP}$ does not saisfy PCA.
	
	Now, consider the $L^{(1)}$ solution, $R_{L^{(1)}}$, then it satisfies PCA, but not EC. In fact, taking the same example used in Statement (ii) it holds  $i P^{\succsim}_{L^{(1)}} j$
	and $j P^{\sqsupseteq}_{L^{(1)}} i$, but, as already observed, a solution that satisfies EC on $\succsim$ and $\sqsupseteq$ should rank $i$ and $j$ in the same way in both coalitional rankings $\succsim$ and $\sqsupseteq$. So, $L^{(1)}$ solution does not satisfy EC.
	
\end{proof}

	Given the two rankings $\succsim,\succsim'$ and  and two players $i,j$, in \cite{bernardi:hal-02191137}, the authors call tha coalitional ranking $\succsim'$ $i$-improving and $j$-invariant with respect to $\succsim$ if the difference between the two rankings relies on the fact that in $\succsim'$ some subsets containing $i$ but not $j$ are strictly better ranked than in $\succsim'$, and no subset containing $j$ has changed its ranking position with respect to $\succsim'$. Moreover they introduce the following axiom 
	\begin{axiom}[ Monotonicity (M) ] \label{Mono}
		A solution $R: \bm{\mathcal{R}}(\PN) \longrightarrow \bm{\mathcal{B}}(N)$ is monotone if for all power relations $\succsim,\succsim'\in \bm{\mathcal{R}}(\PN)$ and  $i, j \in N$  such that $i I^{\succsim}j$ and $\succsim'$ is $i$-improving and $j$-invariant with respect to $\succsim$, then it holds that $i P^{\succsim'} j$.
	\end{axiom}
	
	It is not difficult to prove that 
	\begin{proposition}\,
  \begin{itemize}
      \item If a solution $R:\bm{\mathcal{R}}(\PN) \longrightarrow \bm{\mathcal{B}}(N)$ satisfies Axioms \ref{Sym} (SYM) and \ref{Mono} (M) then $R$ satisfies Axiom \ref{Des2} (SDes).
      \item There exists a solution $R:\bm{\mathcal{R}}(\PN) \longrightarrow \bm{\mathcal{B}}(N)$ satisfying Axioms \ref{Des2} (SDes) and \ref{Sym} (SYM), but not Axiom  \ref{Mono} (M).
  \end{itemize}
	\end{proposition}
	
    The rankig solution lex-cel and $L^{(1)}$ are defined as total binary realtions, but they are actually total preorder because the lexicogrpahic order is transitive.

\subsection{Independence of the axioms for CP-majority}
We prove that Axioms \ref{Des2} (SDes), \ref{N} (Neu), \ref{EC} (EC) and \ref{CAT} (CAT) are independent.\\

\textbf{Strictly desirability is not satisfied:} Consider the solution $R_{Id}$ that ranks all the player equally, i.e. $\forall i,j\in N$, $i I^{\succsim}_{Id} j$ for all $\succsim\in\PP(\PN)$. It easy to observe that this solution satisfies Axioms  \ref{N} (Neu), \ref{EC} (EC) and \ref{CAT} (CAT), but not  Axiom \ref{Des2} (SDes).\\

\textbf{Equality of coalitions is not satisfied:} Consider the solution $R_{EC}$ that coincides with the CP-majority solution if
\begin{itemize}
	\item for a ranking $\succsim\in\PP(\PN)$, the CP-majority produces a total order
	\item for a ranking $\succsim\in\PP(\PN)$ and two players $i,j\in N$ such that  $\emptyset$ does not belong neither to $D_{i,j}^\succsim$ nor $D_{j,i}^\succsim$ and $|D_{i,j}^\succsim|=|D_{j,i}^\succsim|$.
\end{itemize}
In the other cases if $\emptyset\in D_{i,j}^\succsim$ then $i P_{EC}^\succsim j$.  By definition this solution satisfies Axiom  \ref{Des2} (SDes) and \ref{N} (Neu). Moreover given $\succsim,\succsim'\in\PP(\PN)$ such that $i I_{EC}^{\succsim} j$ and $i I_{EC}^{\succsim'} j$ then $i\sim j$ and $i\sim ' j$ then removing a subset of $\succsim\cap\succsim'$ does not affect the ranking between the player, because either we can apply CP-majority or we add the emplty set in both $D_{i,j}^\succsim$ and $D_{i,j}^{\succsim'}$ ( or $D_{j,i}^\succsim$ and $D_{j,i}^{\succsim'}$). Then this solution satisfies Axiom \ref{CAT} (CAT). Consider the rankings $\succsim,\succsim'\in\PP(\PN)$ such that $|D_{i,j}^\succsim|=|D_{j,i}^\succsim|$ and $|D_{i,j}^{\succsim'}|=|D_{j,i}^{\succsim'}|$, $\emptyset\in D_{i,j}^\succsim$, $T\notin D_{j,i}^\succsim$ and  $\emptyset\notin D_{i,j}^{\succsim'}$, $T\in D_{j,i}^{\succsim'}$ then $i P_{EC}^\succsim j$ and $j P_{EC}^\succsim i$. Taking the permutation $\pi(\emptyset)=T $ we observe that $R_{EC}$ does not satisfies Axiom \ref{EC} (EC).\\

\textbf{Neutrality is not satisfied:} Consider a solution $R_N$ that coincides with $R_{CP}^\succsim$ when for a ranking $\succsim\in\PP(\PN)$ the CP-majority produces a total order. In the other cases, given  $\succsim$, it holds  $iP^{\succsim}_Nj$ if and only if $i < j$. This solution satisfies \ref{Des2} (SDes), \ref{EC} (EC) and \ref{CAT} (CAT), but not Axiom \ref{N} (Neu). \\

\textbf{Consistency after tiebreaks is not satisfied:}  Given a power relation $\succsim \in \bm{\mathcal{R}}\left(\PN\right)$ with associated quotient order $\Sigma_1\succ\ldots\succ\Sigma_l$ and two elements $i, j \in N$ we define the sets: $D_{i j}^k(\succsim)=\{S \in$ $\left.2^{N \backslash\{i, j\}}: S \cup\{i\}\in \Sigma_k, S \cup\{j\}\in\Sigma_{k'}, k<k'\right\}$, for $k\in\{1,\ldots,l-1\}$. Consider the solution $R_{CAT}$ such that
\begin{itemize}
	\item $i P_{CAT}^{\succsim} j$ if there exist $\hat{k}\in\{1,\ldots,l-1\}$ such that $|D_{i j}^{\hat{k}}(\succsim)|>|D_{j i}^{\hat{k}}(\succsim)|$ and, if $\hat{k}>1$, $|D_{i j}^k(\succsim)|=|D_{j i}^k(\succsim)|$, for all $k\in\{1,\ldots \hat{k}-1\}$,
	\item $i I_{CAT}^{\succsim} j$ if $|D_{i j}^k(\succsim)|=|D_{j i}^k(\succsim)|$, for all $k\in\{1,\ldots, l-1\}$.
\end{itemize}
This solution satisfies \ref{Des2} (SDes) because, all $k\in\{1,\ldots l-1\}$, either $|D_{i j}^k(\succsim)|=|D_{j i}^k(\succsim)|=0$ or $|D_{i j}^k(\succsim)|>|D_{j i}^k(\succsim)|$, for . Axiom \ref{EC} (EC) is satisfied because the permutation $\pi$ does not change the cardinality of the sets $|D_{i j}^k(\succsim)|, |D_{j i}^k(\succsim)|$. In order to show that this solution satisfies Axiom \ref{N} (Neu) it is enough to observe that permuting $i$ with player $j$ all the relations are reversed. Consider the rankings $\succsim,\succsim'$ with associated quotient orders
\begin{equation*}
	\Sigma_1 \succ \Sigma_2\qquad \Sigma_1' \succ' \Sigma_2' \succ' \Sigma_3'
\end{equation*}
such that $i I_{CAT}^\succsim j$ and $i I_{CAT}^{\succsim'} j $. Suppose that there are coalitions $S,T\subseteq 2^{N\setminus\{i,j\}}$ such that
\begin{align*}
	& S\cup\{i\}\sim S\cup\{j\}\in\Sigma_1 & S\cup\{i\}\sim' S\cup\{j\}\in\Sigma_3'  \\
	& T\cup\{i\}\sim T\cup\{j\}\in\Sigma_2 & T\cup\{i\}\sim' T\cup\{j\}\in\Sigma_2'.
\end{align*}
Take the set $B=\Big\{\big(S\cup\{i\},S\cup\{j\} \big); \big(T\cup\{j\},T\cup\{i\} \big)  \Big\}\subset \succsim\cap\succsim'$ then
\begin{align*}
	& S\in|D_{ji}^1|,\, \,  T\in|D_{ij}^2|,\,  \Rightarrow jP_{CAT}^{\succsim\setminus B} j   \\
	& S\in|D_{ji}^3|,\, \,  T\in|D_{ij}^2|,\,  \Rightarrow iP_{CAT}^{\succsim'\setminus B} i.
\end{align*}
The previous example shows that $R_{CAT}$ does not satisfies Axiom \ref{CAT} (CAT).

\subsection{Independence of the axioms for lex-cel}
We prove that Axioms \ref{Des2} (SDes), \ref{Sym} (SYM), \ref{CA} (CA) and \ref{indep below} (IWS) are independent.\\

\textbf{Strictly desirability is not satisfied:} Take an arbitrary ranking $\succsim\in\PP(\PN)$ with associated quotient order $\Sigma_1\succ\ldots\succ\Sigma_l$ and for all player $i\in N$ take the $l$-dimensional vector $\theta^{\succsim}(i)$ defined for the lex-cel solution. Define the ranking $R_{SD}^{\succsim}$ such that $i I_{SD}^{\succsim} j$ if $\theta^{\succsim}(i)=\theta^{\succsim}(j)$ and $i P_{SD}^{\succsim} j$ if $\exists s: i_k=j_k, k=1,\dots, s-1\; \mbox{ and } i_s<j_s$. By definition $R_{SD}$ does not satisfies Axiom \ref{Des2} (SDes), but, as lex-cel ranking, this solution satisfies Axiom \ref{Sym} (SYM) and \ref{indep below} (IWS). Moreover the permutation $\pi$ in the definition of Axiom \ref{CA} (CA) does not change the entries of the vector $\theta^{\succsim}$ implying that $R_{SD}$ satisfies Axiom \ref{CA} (CA).\\

\textbf{Symmetry is not satisfied:} Consider a solution $R_S$ that coincides with $R_{\ran}^\succsim$ when, for a ranking $\succsim\in\PP(\PN)$, the lex-cel produces a total order. In the other cases, given  $\succsim$, it holds  $iP_S^{\succsim}j$ if and only if $i < j$. This solution satisfies Axioms \ref{Des2} (SDes) and \ref{indep below} (IWS) by definition, moreover when $i I_{\ran}^\succsim j$  the solution depends only on the name of the players and so it satisfies Axiom \ref{CA} (CA) , but not \ref{Sym} (SYM).\\

\textbf{Coalition anonymity is not satisfied:} Consider a solution $R_{CA}$ that coincides with $R_{\ran}^\succsim$ when, for a ranking $\succsim\in\PP(\PN)$, the lex-cel produces a total order. If there are $i, j$ with $\theta^{\succsim}(i)=\theta^{\succsim}(j)=\left(a_1, \ldots, a_l\right)$, define $\hat{k}$ the first integer such that $a_{\hat{k}}>0$. Then $i P_{CA}^{\succsim} j$ if the largest subset $L_i$ containing $i$ in $\Sigma_{\hat{k}}$ contains more elements than the largest subset $L_j$ containing $j$ in $\Sigma_{\hat{k}}$ and $i I_{CA}^{\succsim} j$ if $|L_i|=|L_j|$. By definition $R_{CA}$ satisfies Axioms \ref{Des2} (SDes) and \ref{indep below} (IWS), moreover if, $S\cup\{i\}\sim S\cup\{j\}$, $\forall S\in2^{N\setminus\{i,j\}}$, then  $|L_i|=|L_j|$ and Axiom \ref{Sym} (SYM) is satisfied. Axiom \ref{CA} (CA) is not satisfied because defining a permutation $\pi$ such that $|\pi(L_i)|<|L_j|$ the rank between $i$ and $j$ is exchanged.\\

\textbf{Independence from the worst set is not satisfied:} Take an arbitrary ranking $\succsim\in\PP(\PN)$ with associated quotient order $\Sigma_1\succ\ldots\succ\Sigma_{l+1}$ and for all player $i\in N$ define the $l$-dimensional vector $\theta^{\succsim,+}(i)=(i_1^+,\ldots,i_{l+1}^+)$ where $i_k^+=i_k+i_{k+1}$ for all $k\in\{1,\ldots,l\}$, and $\theta^{\succsim}(i)=(i_1,\ldots,i_k)$ is defined for the lex-cel solution. Define the ranking $R_{IW}^{\succsim}$ such that $i I_{IW}^{\succsim} j$ if $\theta^{\succsim,+}(i)=\theta^{\succsim,+}(j)$ and $i P_{IW}^{\succsim} j$ if $\exists s: i_k^+=j_k^+, k=1,\dots, s-1\; \mbox{ and } i_s^+>j_s^+$. Observe that if $\theta^{\succsim}(i)\geq_L \theta^{\succsim}(j)$ then $\theta^{\succsim,+}(i)\geq_L \theta^{\succsim,+}(j)$ and Axiom \ref{Des2} (SDes) is satisfied. Moreover if $\theta^{\succsim}(i)=\theta^{\succsim}(j)$ then $\theta^{\succsim,+}(i)= \theta^{\succsim,+}(j)$ and Axiom \ref{Sym} (SYM) is satisfied. Finally, given a permutation$\pi$ and the associtated coalitional ranking $\succsim^\pi$ as in the definition of Axiom \ref{CA} (CA), it holds that as $\theta^{\succsim}(i)=\theta^{\succsim^\pi}(i)$ then  $\theta^{\succsim,+}(i)=\theta^{\succsim^\pi,+}(i)$ and Axiom \ref{CA} (CA) is satisfied. Solution $R_{IW}$ does not satisfies Axiom \ref{indep below} (IWS) because in the definition of $i_l^+$ last equivalence class plays a role and taking two different partitions of $\Sigma_{l+1}$ the values of $i_l^+$ and $j_l^+$ can be increased or decreased.\\

\subsection{Independence of the axioms for  $L^{(1)}$}

We prove that Axioms  \ref{Des2} (SDes), \ref{Sym} (SYM), \ref{PCA} (PCA), \ref{indep below} (IWS), \ref{kdesi} ($k$-DD) and \ref{CI} (CI) are independent.\\

\textbf{Strictly desirability is not satisfied:} Take an arbitrary ranking $\succsim\in\PP(\PN)$ with associated quotient order $\Sigma_1\succ\ldots\succ\Sigma_l$.  Consider a solution $R_{SD}$ that coincides with $R_{L^{(1)}}$ when, for a coalitional raking $\succsim$ and players $i,j\in N$, it does not hold that $M_{sk}^{\succsim,i}\geq M_{sk}^{\succsim,j}$, for all $s\in\{1,\ldots,n\}$ and $k\in\{1,\ldots,l\}$. If $M_{sk}^{\succsim,i}= M_{sk}^{\succsim,j}$, for all $s,k$ then $i I_{SD}^\succsim j$. Otherwise Let $\hat{s}$ and $\hat{k}$ be the minimum indexes such that $M_{\hat{s}\hat{k}}^{\succsim,i}> M_{\hat{s}\hat{k}}^{\succsim,j}$ then $j \ P_{SD}^\succsim \ i $. By definition $R_{SD}$ does not satisfies Axiom \ref{Des2} (SDes), but, as $L^{(1)}$  solution , it satisfies Axioms \ref{Sym} (SYM), \ref{indep below} (IWS), \ref{PCA} (PCA) and \ref{kdesi} ($k$-DD). Axiom \ref{CI} (CI) is satisfied by $R_{SD}$ because adding a subset of coalitions from the last equivalence class to the second-last equivalence class or to the union of first $l-1$ equivalence classes the relative ranking between $i$ and $j$ does not change.  \\

\textbf{Symmetry is not satisfied:}  Consider a solution $R_S$ that coincides with $R_{L^{(1)}}^\succsim$ when, for a coalitional ranking $\succsim\in\PP(\PN)$, the $L^{(1)}$  solutionproduces a total order. In the other cases it holds  $iP_S^{\succsim}j$ if and only if $i < j$. This solution satisfies Axioms \ref{Des2} (SDes), \ref{indep below} (IWS) and \ref{kdesi} ($k$-DD) by definition, moreover when $i I_{L^{(1)}}^\succsim j$  the solution depends only on the name of the players and so it satisfies Axiom \ref{CA} (PCA) and \ref{CI} (CI), but not \ref{Sym} (SYM). \\

\textbf{Independence from the worst set is not satisfied:}  Take an arbitrary ranking $\succsim\in\PP(\PN)$ with associated quotient order $\Sigma_1\succ\ldots\succ\Sigma_{l+1}$.For each $i\in N$, define a new matrix $M^{\succsim,+,i}$ such that
\begin{align*}
	& M^{\succsim,+,i}_{s,k} = M^{\succsim,i}_{s,k}+M^{\succsim,i}_{s,k+1},\ \mbox{ if } s\in\{1,\ldots, n\}, k\in\{1,\ldots,l\},
\end{align*}
where $M^{\succsim,i}$ is the matrix defined for the $L^{(1)}$  solution . Define the solution $R^{\succsim}_{L^{(1)},+}$ defined as $R^{\succsim}_{L^{(1)}}$, but taking $M^{\succsim,+,i}$ instead of $M^{\succsim,i}$. Observe that if $M^{\succsim,i}_{s,k} \geq M^{\succsim,j}_{s,k}$ then $M^{\succsim,+,i}_{s,k} \geq M^{\succsim,+,j}_{s,k}$ and Axioms \ref{Des2} (SDes) and \ref{kdesi} ($k$-DD) are satisfied. Moreover if $M^{\succsim,i}_{s,k} = M^{\succsim,j}_{s,k}$ then $M^{\succsim,+,i}_{s,k} = M^{\succsim,+,j}_{s,k}$ and Axiom \ref{Sym} (SYM) is satisfied. Given a permutation $\pi$ and the associated coalitional ranking $\succsim^\pi$  as in the definition of Axiom \ref{PCA} (PCA) it holds $M^{\succsim,i}=M^{\succsim^\pi,i}$ then  $M^{\succsim,+,i}=M^{\succsim^\pi,+,i}$ and Axiom \ref{PCA} (PCA) is satisfied. Given two rankings $\succsim,\succsim'\in\PP(\PN)$ such that $i I^{\succsim}_{L^{(1)},+} j$ and $i I^{\succsim'}_{L^{(1)},+} j$ then $M^{\succsim,+,i}=M^{\succsim,+,j}$ and $M^{\succsim',+,i}=M^{\succsim',+,j}$. Axiom \ref{CI} (CI) is satisfied because the values of  $M^{\succsim,+,i}$ and $M^{\succsim,+,j}$ are not modified by a changing, coherent to the definition of the axiom, in the last equivalence class. Solution $R_{IW}$ does not satisfies Axiom \ref{indep below} (IWS) because in the definition of $M^{\succsim,+,i}_{s,l}$ last equivalence class plays a role and taking two different partitions of $\Sigma_{l+1}$ we can increase or decrease the values of $M^{\succsim,+,i}_{s,l}$ and $M^{\succsim,+,j}_{s,l}$. \\

\textbf{$k$-desirability on dichotomous power relations  is not satisfied:} Take an arbitrary ranking $\succsim\in\PP(\PN)$ with associated quotient order $\Sigma_1\succ\ldots\succ\Sigma_{l}$.  For all player $i\in N$ take the matrix $M^{\succsim,i}$ defined as for the $L^{(1)}$  solution . Define the ranking $R^{\succsim}_{kD}$ such that $$
M^{\succsim, i}=M^{\succsim, j} \Rightarrow \ i I^{\succsim}_{kD} j
$$
and
$$
i \ P^\succsim_{kD} \ j \Longleftrightarrow
\begin{array}{l}
	\exists \hat{s} \in \{1, \ldots, n\} \mbox {and } \hat{k} \in  \{1, \dots, l-1\} \mbox{ such that } \\
\qquad \left\{\begin{array}{ll}
	M^{\succsim, i}_{sk} = M^{\succsim, j}_{sk},                         & \forall s \in \{1, \ldots, n\}, \forall k \in \{1, \ldots, \hat{k}-1\}; \smallskip \\
	M^{\succsim, i}_{s\hat{k}} = M^{\succsim, j}_{s\hat{k}},             & \forall s \in \{\hat{s}+1,\ldots,n\}; \smallskip                                   \\
	M^{\succsim, i}_{\hat{s}\hat{k}} < M^{\succsim, j}_{\hat{s}\hat{k}}. &
\end{array}\right.
\end{array}
$$
    This solution satisfies, by definition, Axioms  \ref{Des2} (SDes) and \ref{Sym} (SYM) because if $S\cup\{i\}\succsim S\cup\{j\}$, $\forall S\in2^{N\setminus\{i,j\}}$, then $M^{\succsim, i}_{sk} \geq M^{\succsim, j}_{sk}$ \ref{Sym} (SYM), for all indeces $s$ and $k$. Axioms \ref{indep below} (IWS) and \ref{PCA} (PCA) are satisfied because the ranking does not depend on the equivalence classes $\Sigma_k$, with $k>\hat{k}$, and because the values of the matrices $M^{\succsim, i}, M^{\succsim, i}$ does not change with a size invariant permutation $\pi$.  Solution $R_{kD}$ satisfies Axiom \ref{CI} (CI) because in new rankings  $\succsim',\succsim''$ the values of the matrices $M^{\succsim, i}$, $M^{\succsim, j}$ and $M^{\succsim', i}$, $M^{\succsim', j}$ are modified by the same quantity.  Axiom \ref{kdesi} ($k$-DD) is not satisfied because if $M^{\succsim, i}_{s\hat{k}} < M^{\succsim, j}_{s\hat{k}}$, for some $s<\hat{s}$, and according to the axiom it should be $j P^{\succsim}_{kD} i$. \\

\textbf{Per-size Coalitional Anonymity is not satisfied:} Take an arbitrary ranking $\succsim\in\PP(\PN)$ with associated quotient order $\Sigma_1\succ\ldots\succ\Sigma_l$. Define, for each player $i\in N$, the set $E_i^k=\{j\in N\setminus\{i\}:\{i,j\}\in\Sigma_k\}$. Take a solution $R_{PCA}$ that coincides with  $R_{L^{(1)}}^{\succsim}$ when the $L^{(1)}$  solutionproduces a total order. Given two players $i,j\in N$ such that $i I_{L^{(1)}}^{\succsim} j $,  let $\hat{k}$ the first index such that if $\min\{ E_i^{\hat{k}}\} \neq \min \{ E_j^{\hat{k}}\}$ then $i P^{\succsim}_{PCA} j$ if and only if $\min\{ E_i^{\hat{k}}\} < \min \{ E_j^{\hat{k}}\}$ and $i I^{\succsim}_{PCA} j$, otherwise. By definition this solution satisfies Axioms \ref{Des2} (SDes), \ref{Sym} (SYM) and \ref{kdesi} ($k$-DD), but not \ref{PCA} (PCA). Observe that, when $L^{(1)}$ does not apply, only the equivalence class $\Sigma_{\hat{k}}$ is involved to determine the raking of players then Axiom \ref{indep below} (IWS) is satisfied. Solution $R_{PCA}$ satisfies Axiom \ref{CI} (CI) because if $i I^{\succsim}_{PCA} j$ then that $\min\{ E_i^{k}\} = \min \{ E_j^{k}\}$, for every $k$, then adding a subset of coalitions from the last equivalence class to the second-last equivalence class or to the union of the first $l-1$ equivalence classes the relative ranking between $i$ and $j$ does not change.\\

\textbf{Consistency after indifference is not satisfied:} Take an arbitrary ranking $\succsim\in\PP(\PN)$ with associated quotient order $\Sigma_1\succ\ldots\succ\Sigma_l$. Take a solution $R_{CI}$ that coincides with  $R_{L^{(1)}}^{\succsim}$ when the $L^{(1)}$  solutionproduces a total order. Denote by $i_k^{\min}$ the number of sets in $\Sigma_{k}$ containing $i$ with minimal cardinality, i.e. 
\[ i_k^{\min} = |\{S\in \Sigma_k: i\in S, |S|<|G|, \mbox{ where } i\in G\in\Sigma_k\}| \]
for $k=1, \dots,l$. Let $\theta^\succsim_{\min}(i)$ be the $l$-dimensional vector $\theta^\succsim(i)=(i_1^{\min},\dots,i_l^{\min})$ associated to $\succsim$. Take two players $i,j\in N$ such that $i I_{L^{(1)}}^{\succsim} j$ then
\begin{align*}
& \mbox{ if } \theta^\succsim_{\min}(i)\ge_L \theta^\succsim_{\min}(j) \mbox{ then } i P^{\succsim}_{CI} j; \\
& \mbox{ otherwise } i I^{\succsim}_{CI} j.
\end{align*}
This solution satisfies, by definition, Axioms  \ref{Sym} (SYM), \ref{indep below} (IWS) and \ref{kdesi} ($k$-DD). Moreover vector $\theta^\succsim_{\min}$ depends only on the number of coalition of minimal cardinality containing a player then Axiom \ref{PCA} (PCA) is satisfied. Solution $R_{CI}$ does not satisfies Axiom \ref{CI} (CI) because a minimal set in the second equivalence class of a ranking could be not minimal in the union of the last $l-1$ equivalence classes.\\

The following proposition concludes the axiomatic analysis presented in this work.
\begin{proposition}\label{completeaxioms}
The following statements  hold true:
\begin{itemize}
\item[(i)] lex-cel, \emph{dual-lex}, $L^{(1)}$ solution and $L_*^{(1)}$ solution satisfy Axiom \ref{N} (Neu);
\item[(ii)] \emph{dual-lex} and $L_*^{(1)}$ solution do not satisfy Axiom \ref{EC} (EC) ;
\item[(iii)] lex-cel, \emph{dual-lex}, $L^{(1)}$ solution and $L_*^{(1)}$ solution satisfy Axiom \ref{CAT} (CAT) ;
\item[(iv)] $L^{(1)}$ solution and $L_*^{(1)}$ solution do not satisfy Axiom \ref{CA} (CA) ;
\item[(v)] CP-majority, \emph{dual-lex} and $L_*^{(1)}$ solution do not satisfy Axiom \ref{indep below} (IWS) ;
\item[(vi)] CP-majority, lex-cel and $L^{(1)}$ solution do not satisfy Axiom \ref{indep best} (IWS) ;
\item[(vii)] CP-majority, lex-cel and \emph{dual-lex} do not satisfy Axiom \ref{kdesi} ($k$-DD);
\item[(viii)]CP-majority, does not satisfy Axiom \ref{CI} (CI);
\item[(ix)] lex-cel and \emph{dual-lex} both satisfy Axiom \ref{CI} (CI);
\end{itemize}
\end{proposition}
\begin{proof}

\noindent
Proof of statement (i)\\
The proof is straigthforward, and can also be found in \cite{bernardi:hal-02191137}, for lex-cel and \emph{dual-lex}, and in \cite{algaba:hal-03388789}, for $L^{(1)}$ solution and $L_*^{(1)}$ solution.\\

\noindent
Proof of statement (ii)\\
Use the same example of the proof of statement (ii) and (iii) of Proposition \ref{logic1} showing that neither \emph{dual-lex} nor $L^{(1)}$ solution satisfy EC.\\

\noindent
Proof of statement (iii)\\
We prove that lex-cel satisfies CAT.

First, recall that, in the statement of CAT,  removing the same set $B\subseteq\succsim\cap\succsim'$ from $\succsim$ and $\succsim'$ such that $\succsim\setminus B$ and $\succsim'\setminus B$ remain total preorders, means that one is removing from $\succsim$ and $\succsim'$ some pairs  $(S,T) \in \succsim\cap\succsim'$ such that there exist a corresponding equivalence class for each quotient order $\Sigma$ with $S,T \in \Sigma$ in both quotient orders of $\succ$ and $\succ'$. So, for any $(S,T) \in B$ it must exist an equivalence class $\Sigma$ such that (see Remark \ref{rem:cat})
$$
\Sigma_1  \succ \ldots \succ \Sigma \succ \ldots \succ \Sigma_l
$$
and
$$
\Sigma'_1 \succ' \ldots \succ \Sigma \succ' \ldots \succ' \Sigma'_m,
$$
such that after removal of $(S,T) \in B$ we obtain two quotient orders $\hat{\succ}=\succ\setminus B$ and $\hat{\succ}'=\succ'\setminus B$
$$
\Sigma_1 \hat{\succ} \ldots \hat{\succ} \Sigma\setminus \Omega\ \hat{\succ}\ \Omega\  \hat{\succ} \ldots \hat{\succ} \Sigma_l,
$$
and
$$
\Sigma'_1 \hat{\succ}' \ldots \hat{\succ}' \Sigma\setminus \Omega \ \hat{\succ}' \ \Omega \  \hat{\succ}' \ldots \hat{\succ}' \Sigma'_m,
$$
with $S \in \Omega$ and $T \in \Sigma \setminus \Omega$. If $i,j \in N$ are such that $i I_{\ran}^{\succsim} j$ and $i I_{\ran}^{\succsim'} j$, we have that   $\theta^\succsim(i) = \theta^\succsim(j)$ and $\theta^{\succsim'}(i) = \theta^{\succsim'}(j)$. Now, after $B$ is removed, two cases may arise: case (1) each  $\Sigma\setminus \Omega$ and $\Omega$ as above contain a same number of coalitions containing $i$ or $j$, and in this case we have  $\theta^{\hat{\succsim}}(i) = \theta^{\hat{\succsim}}(j)$ and $\theta^{\hat{\succsim}'}(i) = \theta^{\hat{\succsim}'}(j)$; case (2) $\Sigma\setminus \Omega$ and $\Omega$ contain a different number of coalitions containing $i$ or $j$, but in this case the highest equivalence class $\Sigma\setminus \Omega$ showing a different number of coalitions for $i$ and $j$ will be the same in both $\hat{\succ}$ and $\hat{\succ}'$, so $\theta^{\hat{\succsim}}(i) >_L \theta^{\hat{\succsim}}(j) \Leftrightarrow \theta^{\hat{\succsim}'}(i) >_L \theta^{\hat{\succsim}'}(j)$. In both cases (1) and (2), we have that $i R_{\ran}^{\succsim} j \Leftrightarrow  R_{\ran}^{\succsim'} j$, which is precisely what is demanded by CAT.

The proof that \emph{dual-lex}, $L^{(1)}$ solution and $L_*^{(1)}$ solution satisfy CAT follows similar lines and it is left to the reader.\\

\noindent
Proof of statement (iv)\\
Use the same example of the proof of statement (ii) of Proposition \ref{logic1} showing that neither CP-majority, $L^{(1)}$ solution nor $L_*^{(1)}$ solution satisfy CA.\\

\noindent
Proof of statement (v)\\
It is obvious that CP-majority does not satisfy IWS. To see that \emph{dual-lex} does not satisfy IWS we refer to Section 6 in \cite{bernardi:hal-02191137}. The proof that neither $L_*^{(1)}$ solution satisfies IWS follows a similar path.\\

\noindent
Proof of statement (vi)\\
Similar to case (v).\\

\noindent
Proof of statement (vii)\\
Consider a coalitional ranking $\succsim \in \PP(\PN)$ such that
\[
i \sim jk \sim jl \succ j \sim ik \sim il \succ \mathcal{C}, 
\]
where $\mathcal{C}$ is an equivalence class containing all the other coalitions not previously shown.
We have that $i P^\succsim_1 j$ and $j P^\succsim_2 i$, while $i I^\succsim_k j$ for all the other $k$. So, a solution satisfying $k$-DD should rank $i$ strictly more relevant than $j$. It is easy to check that on $\succsim$  CP-majority, lex-cel and \emph{dual-lex} all rank $j$ strictly more relevant than $i$.\\

\noindent
Proof of statement (viii)\\
Consider a coalitional ranking $\succsim \in \PP(\PN)$ such that
\[
jk \sim jl \succ ik \sim il \sim i \sim  ikl \succ \mathcal{C}, 
\]
where $\mathcal{C}$ is an equivalence class containing all the other coalitions not previously shown.
One can verify that the CP-majority is such that $i I^{\succsim}_{CP} j$
Now, consider two new coalitional rankings $\succsim^{'}$ and $\succsim^{''}$ such that
\[
jk \sim jl \succ^{'} ik \sim^{'} il \sim^{'} i \sim^{'}  ikl \sim^{'} j  \succ \mathcal{C}\setminus j, 
\]
and
\[
jk \sim^{''} jl \sim^{''} ik \sim^{''} il \sim^{''} i \sim^{''}  ikl\sim^{''} j \succ \mathcal{C}\setminus j, 
\]
where, following the  notation of Axiom \ref{CI} (CI), $\Sigma=\{j\}$. Notice that  $j P^{\succsim^{'}}_{CP} i$ while $i P^{\succsim^{''}}_{CP} j$, contradicting what is demanded by CI. So the CP-majority does not satisfy CI.

\noindent
Proof of statement (ix)\\
The proof that lex-cel satisfies CI is similar the the proof that  $L^{(1)}$ satisfies CI as in Theorem \ref{teo:L1}, and then it is omitted. Similar for \emph{dual-lex}.

\end{proof}

\begin{table}[t]
\begin{center}
\begin{tabular}{|c|c|c|c|c|c|}
\hline
 & CP-majority & lex-cel & dual-lex & $L^{(1)}$ & $L_*^{(1)}$\\
 \hline
SDes & \cmark$^+$ & \cmark$^+$ & \cmark$^+$ &\cmark$^+$ & \cmark$^+$ \\
\hline
Neu &  \cmark$^+$ & \cmark & \cmark & \cmark & \cmark \\
\hline
CAT &  \cmark$^+$ & \cmark & \cmark & \cmark & \cmark \\
\hline
Sym &  \cmark & \cmark$^+$ & \cmark$^+$ & \cmark$^+$ & \cmark$^+$ \\
\hline
EC &  \cmark$^+$ & \xmark & \xmark & \xmark & \xmark \\ 
\hline
CA & \xmark & \cmark$^+$ & \cmark$^+$ & \xmark & \xmark \\
\hline
IWS & \xmark & \cmark$^+$ & \xmark & \cmark$^+$ & \xmark \\
\hline
IBS & \xmark & \xmark & \cmark$^+$ & \xmark & \cmark$^+$ \\
\hline
$k$-DD & \xmark & \xmark & \xmark & \cmark$^+$ & \cmark$^+$ \\
\hline
PCA & \xmark & \cmark & \cmark & \cmark$^+$ & \cmark$^+$ \\
\hline
CI & \xmark & \cmark & \cmark & \cmark$^+$ & \cmark$^+$ \\
\hline
\end{tabular}
\caption{Properties (on the rows) satisfied (\cmark) or not (\xmark) by each social ranking (on the columns). Notation $^+$ highlights sets of properties used in this work to axiomatically characterize the corresponding solutions on each column.}
\label{tab:axiomsandrankings}
\end{center}
\end{table}
\noindent
Table \ref{tab:axiomsandrankings} summarizes which axioms are satisfied by each social ranking studied in this paper.

\section{Case study}\label{sec:cases}

The States General of the Netherlands is the supreme bicameral legislature of the Netherlands.
The House of Representatives (Tweede Kamer) consists of 150 members. They have the ability to propose and to revise a bill.
If a majority votes in favor of a proposed bill, it is moved to the 75 members of the Senate (Eerste Kamer), who only get to vote on the bill.
If a majority of the Senate votes in favor, the bill is passed to law.

The distribution of votes between each party in each house is shown in Table \ref{table:netherlands}, with $w^\ell_i$ and $w^u_i$ corresponding to the number of votes each party $i$ has in the House of Representatives (Lower House) and the Senate (Upper House), respectively.

In determining the power relation between the individual parties based on their contribution to coalitions, Musegaas et al. have proposed to consider the Shapley value of a three-valued simple game, see \cite{musegaas2018three}, which aims to incorporate the asymmetric structure of the vote distribution.
Their value function of the three-valued simple game generates a $2$ for coalitions that are able to pass a legislation in both the House of Representatives and the Senate, a $1$ if it only passes the House of Representatives, and a $0$ otherwise.

As argued in Example \ref{ex2.1}, quantifying the utility of a certain outcome could appear rather arbitrary.
Since merely the relative power between parties is of interest, we may also only consider the relative power between coalitions.

Let $N$ be the set of 17 parties in the States General of the Netherlands. Consider a power relation $\succsim \in \PP(\PN)$ consisting of three equivalence classes.

\begin{align*}
  \Sigma_1 & = \{S \in \PN: \textstyle{\sum_{i \in S}} w^\ell_i \geq 76\ \text{ and }\ \sum_{i \in S} w^u_i \geq 38\} \\
  \Sigma_2 & = \{S \in \PN: \textstyle{\sum_{i \in S}} w^\ell_i \geq 76\ \text{ and }\ \sum_{i \in S} w^u_i < 38\}    \\
  \Sigma_3 & = \PN \setminus (\Sigma_1 \cup \Sigma_2)
\end{align*}

Given the distribution of all coalitions between these three equivalence classes, Table \ref{table:netherlands} shows the resulting ranking of each social ranking solution introduced in previous sections, along with the ranking produced in \cite{musegaas2018three}.
The source code for arriving at those rankings is shown in Listing \ref{lst:helper} in the Appendix.

\begin{table}[t]
  \centering
  \begin{tabular}{lrr|rrrr}
    \hline
    & $w^\ell_i$ & $w^u_i$ & $R_{le}$ & $R_{L^{(1)}}$ & $R_{CP}$ & $\Phi(v)$ \\ 
    \hline
    VVD       & 40 & 13 & 1.  & 1.  & 1.  & 1.  \\
    PvdA      & 36 & 8  & 2.  & 2.  & 2.  & 2.  \\
    SP        & 15 & 9  & 4.  & 4.  & 3.  & 3.  \\
    CDA       & 13 & 12 & 3.  & 3.  & 4.  & 4.  \\
    D66       & 12 & 10 & 5.  & 5.  & 5.  & 5.  \\
    PVV       & 12 & 9  & 6.  & 6.  & 6.  & 6.  \\
    CU        & 5  & 3  & 8.  & 8.  & 7.  & 7.  \\
    GL        & 4  & 4  & 7.  & 7.  & 8.  & 8.  \\
    SGP       & 3  & 2  & 9.  & 9.  & 9.  & 9.  \\
    PvdD      & 2  & 2  & 10. & 10. & 10. & 10. \\
    GrKO      & 2  & 0  & 12. & 13. & 12. & 12. \\
    GrBvK     & 2  & 0  & 12. & 13. & 12. & 12. \\
    50PLUS    & 1  & 2  & 11. & 11. & 11. & 11. \\
    Houwers   & 1  & 0  & 15. & 15. & 14. & 14. \\
    Klein     & 1  & 0  & 15. & 15. & 14. & 14. \\
    Van Vliet & 1  & 0  & 15. & 15. & 14. & 14. \\
    OSF       & 0  & 1  & 14. & 12. & 17. & 17. \\
    \hline
  \end{tabular}
  \caption{Party votes and rankings according to the lex-cel, $L^{(1)}$, CP and $\Phi(v)$\cite{musegaas2018three}}
  \label{table:netherlands}
\end{table}

From the evaluation, we observe some minor but notable distinctions in how different social ranking solutions prioritize the influence of parties in the bicameral system. Specifically, both the lex-cel and $L^{(1)}$  solution s exhibit a marked preference for parties with stronger representation in the Senate. This preference is especially evident when comparing the rankings of the parties Van Vliet and OSF. Despite Van Vliet having a solitary vote in the House of Representatives and OSF having one in the Senate, both lex-cel and $L^{(1)}$ solution  assign a higher rank to OSF, diverging from the ranking provided by CP-majority.

Delving deeper, we find that $L^{(1)}$ solution exhibits a more pronounced bias towards smaller coalitions in the Senate compared to lex-cel.
When compared to the party GrKO, which holds twice as many votes in the House of Representatives as Van Vliet, $L^{(1)}$ solution  continues to favor OSF, unlike lex-cel, which prefers GrKO.
The lex-cel vectors and $L^{(1)}$ matrices for these two parties in (\ref{eq:nlvector}) and (\ref{eq:nlmatrix}) show that this ranking difference stems from the small difference in how many winning coalitions of size 5 OSF can participate in, namely $M^{\succsim,\text{OSF}}_{5,1}=13$ over the 12 by GrKO.
In this context, we could argue that, since a bill is only passed if approved by both houses, and smaller coalitions are generally easier to form than bigger ones, the ranking $L^{(1)}$ solution is an appropriate choice to highlight the influence of parties with a higher representation in the Senate.

\begin{equation}\label{eq:nlvector}
  \theta^\succsim(\text{GrKO}) = (\bm{28838}, 4442, 32256)\quad
  \theta^\succsim(\text{OSF}) = (\bm{28747}, 3517, 33272)
\end{equation}


\begin{equation}\label{eq:nlmatrix}
M^{\succsim,\text{GrKO}} = \begin{pmatrix}
   0 &    0 &    1 \\
   0 &    0 &   16 \\
   0 &    1 &  119 \\
   0 &   14 &  546 \\
  \bm{12} &   86 & 1722 \\
 161 &  307 & 3900 \\
 910 &  696 & 6402 \\
\vdots & \vdots & \vdots \\
   1 &    0 &    0 
\end{pmatrix} \quad
M^{\succsim,\text{OSF}} = \begin{pmatrix}
   0 &   0 &    1 \\
   0 &   0 &   16 \\
   0 &   1 &  119 \\
   0 &  14 &  546 \\
  \bm{13} &  84 & 1723 \\
 159 & 284 & 3925 \\
 886 & 581 & 6541 \\
\vdots & \vdots & \vdots \\
   1 &   0 &    0
\end{pmatrix}
\end{equation}


Lastly, the CP-majority relation produces a transitive relation which is remarkably identical to the Shapley value ranking as computed in \cite{musegaas2018three}. As demonstrated in Example \ref{ex3.1CP}, we can construct a matrix displaying the CP-comparisons for each party. Here, we focus on the three parties GrKO, Van Vliet and OSF, which produce a different ranking in all three analyzed ranking methods.

\begin{equation}
\begin{array}{lccc}
i\quad \backslash\quad j& \text{GrKO} & \text{Van Vliet} & \text{OSF} \\
\text{GrKO}      &   0 &   512 & 1016 \\
\text{Van Vliet} &   0 &     0 &  504 \\
\text{OSF}       & 405 &   437 &    0 \\
\end{array}
\end{equation}

\section{Final remarks}\label{sec:final}

In this paper, we present an axiomatic analysis of five social rankings from the literature: CP-majority \cite{haret:hal-02103421}, lex-cel (and its dual) \cite{bernardi:hal-02191137}, and the $L^{(1)}$ solution (and its dual) \cite{algaba:hal-03388789}. Our aim is to better understand the ability of these social rankings to assess the importance of individuals or elements in scenarios where coalition outcomes 
are assessed using an ordinal scale.
Our findings reveal a common basis in their axiomatic foundations, primarily consisting in the property of Strict Desirability  (a key focus of this study at the crossroad between social choice and coalitional games \cite{I58,taylor2021simple,carreras1996complete}) alongside Neutrality (or the weaker Symmetry axiom) and the property of Consistency After Tiebreaks. 

We emphasize the fact that our axiomatic study brings new insights on those social rankings, which appear closer than anticipated. Specifically, we find that the CP-majority differs from lex-cel primarily due to alternative combinations of three simple axioms: Equivalence of Coalitions for the CP-majority, and Coalitional Anonymity plus Independence of the Worst Set for the lex-cel, with all other axioms in our characterizations remaining consistent. Essentially, this implies that the CP-majority distinguishes itself from lex-cel by its ability to remain independent of the positions of pairs of CP-comparisons $S \cup \{i\}$ {\it vs.} $S \cup \{j\}$ for all $i,j$ in $N$ with $S \in 2^{N\setminus \{i,j\}}$ (Axiom EC), while for lex-cel, positions in the coalitional rankings hold greater significance (Axiom CA), except for those in the worst set (Axiom IWS).
On the other hand, due to its prominent lexicographic nature, the $L^{(1)}$ solution bears more similarity to lex-cel, differing by a weaker version of Axiom CA (the Axiom PCA) and the use of a property of  cardinality-based desirability on dichotomous coalitional rankings (Axiom $k$-DD), which is essential for its ability to discriminate based on the coalitional size.

Overall, our analysis shows that CP-majority, lex-cel, and $L^{(1)}$ solution, in this order, exhibit increasing levels of sensitivity in considering the position and the structure of coalitions through the coalitional ranking. This conclusion, derived from the interpretation of the axioms, broadly aligns with the results obtained from the case study in Section \ref{sec:cases}. In fact, the ranking provided by the CP-majority aligns identically with the one of the Shapley value, which is essentially based on an average marginal contribution to coalition's values and is only weakly correlated with the ordinal position of coalitions. Instead, the differences are more pronounced between lex-cel and the Shapley value, and even more so between $L^{(1)}$ and the Shapley value, particularly for individuals of lesser excellence.
In view of these evidences, our main suggestion for guiding the selection of the most suitable social ranking would be the gradient of importance that the ordinal position of coalitions and their structure should hold, ranging from CP-majority (low level) to lex-cel (intermediate level) and finally the $L^{(1)}$ solution (high level), when the reward of excellence is the primary criterion for measuring the relevance of elements (or the penalty of mediocrity if one considers the dual versions, in the same order).

As for new research directions, we believe it would be highly beneficial to characterize the family of social rankings that satisfy the axioms common to all those studied in this paper. Additionally, an open question remains regarding which set of independent axioms should be imposed to ensure that each total binary relation in the image of a social ranking is also transitive, as it is the case for lex-cel and $L^{(1)}$, but not for CP-majority, in general.


\appendix
\section{Appendix}
\subsection{Calculating the scores of the bicameral legislature in the Netherlands}

For the computational tasks, the code predominantly leverages base R functionalities.
However, for specialized operations, we employ the \texttt{socialranking} package\cite{fritz23}, which provides essential functions like \texttt{createPowerset()} for generating the power set $2^N$ of all possible coalitions and \texttt{doRanking()} to determine the ranking in this legislative context.
It should be noted that the package also incorporates the lex-cel, $L^{(1)}$, and CP-majority social ranking functions.

\begin{lstlisting}[style=Rstyle, caption={Applying the social ranking solutions to the States General in the Netherlands}, label=lst:helper]
library(socialranking)

weights <- matrix(
  c(
    40, 36, 15, 13, 12, 12, 5, 4, 3, 2, 2, 2, 1, 1, 1, 1, 0,
    13,  8,  9, 12, 10,  9, 3, 4, 2, 2, 0, 0, 2, 0, 0, 0, 1
  ) |> as.integer(),
  ncol = 2,
  dimnames = list(
    c('VVD', 'PvdA', 'SP', 'CDA', 'D66', 'PVV', 'CU', 'GL', 'SGP', 'PvdD', 'GrKO', 'GrBvK', '50PLUS', 'Houwers', 'Klein', 'Van Vliet', 'OSF'),
    c('House 1', 'House 2')
  )
)

powerset <- createPowerset(1:nrow(weights))

# 2^n x 2 matrix (two columns for house 1 and house 2 weights)
coalitionWeights <- sapply(powerset,
  function(coalition) apply(weights[coalition,,drop=FALSE], 2, sum)
)

qs <- apply(weights, 2, sum) %/% 2 + 1
determineEqCl <- function(votes) {
  if(votes['House 1'] < qs['House 1']) 3 else
  if(votes['House 2'] < qs['House 2']) 2 else 1
}

# Lexcel
lexMatrix <- matrix(0, ncol=ncol(weights)+1, nrow=nrow(weights))
dimnames(lexMatrix) <- list(rownames(weights),paste0('E',0:ncol(weights)+1))
for(i in 1:length(powerset)) {
  eqcl <- determineEqCl(coalitionWeights[,i])
  lexMatrix[powerset[[i]], eqcl] <- lexMatrix[powerset[[i]], eqcl] + 1
}
(lexRank <- lapply(1:nrow(lexMatrix), function(i) lexMatrix[i,])
  |> structure(class='LexcelScores', names=rownames(weights))
  |> doRanking()
)

# L1
l1Matrix <- array(0, c(nrow(weights),nrow(weights),ncol(weights)+1))
dimnames(l1Matrix) <- list(rownames(weights), 1:nrow(weights), paste0('E', 0:ncol(weights)+1))
for(i in 1:length(powerset)) {
  coal <- powerset[[i]]
  eqcl <- determineEqCl(coalitionWeights[,i])
  l1Matrix[coal,length(coal),eqcl] <- l1Matrix[coal,length(coal),eqcl] + 1
}
(l1Rank <- lapply(1:nrow(weights), function(i) l1Matrix[i,,])
  |> structure(class='L1Scores', names=rownames(weights))
  |> doRanking()
)

# CP
cpMatrix <- structure(matrix(0, nrow=nrow(weights), ncol=nrow(weights)), dimnames=list(rownames(weights), rownames(weights)))

for(i in 1:(nrow(weights)-1)) {
  for(j in (i+1):nrow(weights)) {
    
    Uij <- createPowerset(setdiff(1:nrow(weights), c(i,j)))
    eq_ij <- sapply(Uij,
      function(S) c(apply(weights[c(S,i),,drop=FALSE], 2, sum))
    ) |> apply(2, determineEqCl)
    eq_ji <- sapply(Uij,
      function(S) c(apply(weights[c(S,j),,drop=FALSE], 2, sum))
    ) |> apply(2, determineEqCl)
    d_ij <- sum(eq_ij < eq_ji)
    d_ji <- sum(eq_ji < eq_ij)
    cpMatrix[i,j] <- d_ij
    cpMatrix[j,i] <- d_ji
  }
}
\end{lstlisting}

\clearpage
\bibliographystyle{abbrv}
\bibliography{references}

\end{document}